\documentclass{eptcs}
 \providecommand{\firstpage}{83}

\usepackage{breakurl,hyperref}

\usepackage{amsmath,amsfonts,amssymb,amsthm,stmaryrd}%
\usepackage{branesymb,proof,enumitem,bussproofs,caption,graphicx}

\DeclareMathAlphabet{\mathpzc}{OT1}{pzc}{m}{it}

\newcommand{\N}{\mathbb{N}}

\newcommand{\R}{\mathbb{R}}
\newcommand{\B}{\mathbb{B}}

\newcommand{\A}{\mathbb{A}}

\newcommand{\set}[1]{\{ #1 \}}
\newcommand{\trs}{\xrightarrow}

\newcommand{\rlabel}[1]{\; \text{\scriptsize{(#1)}}}

\newcommand{\mabc}[1]{\llbracket #1 \rrbracket}
\newcommand{\bamc}[1]{\llparenthesis #1 \rrparenthesis}
\newcommand{\Mid}{\; \mid \;}
\newcommand{\tuple}[1]{\langle #1 \rangle}

\newcommand{\mpar}[2]{{}_{#1}\kern-2pt\varobar_{#2}}
\newcommand{\mtimes}[2]{{}_{#1}\kern-1pt\otimes_{#2}}
\newcommand{\mAt}[2]{@^{#1}_{#2}}

\newcommand{\sys}{\mathsf{sys}}
\newcommand{\mem}{\mathsf{mem}}

\newcommand{\ph}{\text{\rm ph}}
\newcommand{\coph}{\ph^{\bot}}
\newcommand{\ex}{\text{\rm ex}}

\newcommand{\mset}{\mathpzc}

{\medskip\par\noindent%
\addtolength{\abovedisplayskip}{-0.3\abovedisplayskip}%
\addtolength{\belowdisplayskip}{-0.6\belowdisplayskip}%
\rule{1pt}{4pt}\leaders\hrule depth-3pt height4pt\hfill%
~\textsf{\textbf{#1}}~%
\leaders\hrule depth-3pt height4pt\hfill\rule{1pt}{4pt}%
\nopagebreak}%
{\nopagebreak\rule{1pt}{4pt}\leaders\hrule height1pt\hfill\rule{1pt}{4pt}
\newline}

\newtheoremstyle{theorem}%
{\topsep}%
{\topsep}%
{\it}%
{}%
{\bfseries}%
{.}%
{0.5em}%
{\thmname{#1}\thmnumber{ #2}\thmnote{ (#3)}}%
\theoremstyle{theorem} %
\newtheorem{theorem}{Theorem}[section]
\newtheorem{proposition}[theorem]{Proposition}
\newtheorem{lemma}[theorem]{Lemma}

\newtheorem{definition}[theorem]{Definition}

\title{Implementing the Stochastics Brane Calculus\\ in a
  Generic Stochastic Abstract Machine%
}
\author{Marino Miculan \qquad Ilaria Sambarino
\institute{Department of Mathematics and Computer Science, University of Udine, Italy}
\email{marino.miculan@uniud.it \qquad\qquad ilaria.sambarino@gmail.com}
}
\def\titlerunning{Implementing Stochastics Brane Calculus in COWGSAM}
\def\authorrunning{M. Miculan and I. Sambarino}

\begin{document}
\maketitle

\begin{abstract}
  In this paper, we deal with the problem of implementing an abstract
  machine for a stochastic version of the Brane Calculus.  Instead of
  defining an \emph{ad hoc} abstract machine, we consider the generic
  stochastic abstract machine introduced by Lakin, Paulev\'e and
  Phillips. The nested structure of membranes is flattened into a set
  of species where the hierarchical structure is represented by means
  of \emph{names}.  In order to reduce the overhead introduced by this
  encoding, we modify the machine by adding a \emph{copy-on-write}
  optimization strategy.  We prove that this implementation is
  adequate with respect to the stochastic structural operational
  semantics recently given for the Brane Calculus.  These techniques
  can be ported also to other stochastic calculi dealing with nested
  structures.
\end{abstract}

\section{Introduction}

A fundamental issue in Systems Biology is modelling the membrane
interaction machinery.  Several models have been proposed
in the literature \cite{lt:biokappa,rpscs:bioambients,bmsmt:cls}; among
them, the Brane Calculus \cite{cardelli04:bc} has been arisen as
a good model focusing on abstract membrane interactions, still being
sound with respect to biological constraints (e.g.~bitonality).  In
this calculus, a process represents a system of nested compartments,
where active components are \emph{on} membranes, not inside them.
This reflects the biological evidence that functional molecules
(proteins) are embedded in membranes, with consistent orientation.

In the original definition of the Brane Calculus \cite{cardelli04:bc}
(which we will recall in Section~\ref{sec:bc}) membranes interact
according to three basic reaction rules corresponding to
\emph{phagocytosis}, \emph{endo/exocytosis}, and \emph{pinocytosis}.
However, this semantics does not take into account quantitative
aspects, like stochastic distributions, which are important for, e.g.,
implementing stochastic simulations.

A stochastic semantics for the Brane Calculus has been provided in
\cite{bm:tcs12}, following an approach pioneered in \cite{cm:quest10}
(but see also \cite{hermanns02,prakashbook09} for Markov processes).
Instead of giving a stochastic version $P\trs{a,r} Q$ of the reaction
relation, in this semantics each process is given a \emph{measure} of
the stochastic distribution of the possible outcomes.  More precisely,
we define a relation $P\to \mu$ associating to a process $P$ an
action-indexed family of measures $\mu$: for an action $a$, the
measure $\mu_a$ specifies for each measurable set $S$ of processes,
the rate $\mu_a(S)\in \R^+$ of $a$-transitions from $P$ to (elements
of) $S$.  An advantage of this approach is that we can apply results
from measure theory for solving otherwise difficult issues, like
instance-counting problems; moreover, process measures are defined
\emph{compositionally}, and in fact the relation $P\to \mu$ can be
characterized by means of a set of rules in a GSOS-like format. We
will recall this stochastic semantics and its main properties in
Section~\ref{sec:sgsos}.

In this paper, we use this new semantics for defining a
\emph{stochastic abstract machine} for the Brane Calculus, so that it
can be effectively used for \emph{in silico} simulations of membrane
systems. Defining an \emph{ad hoc} abstract machine for the Brane
Calculus would be a complex task; instead, we take advantage of the
\emph{generic abstract machine for stochastic process calculi} (GSAM
for short) introduced in \cite{pylp:cmsb10,lpp:tcs12} as a general
tool for simulating a broad range of calculi.  This machine can be
instantiated to a particular calculus by defining a function for
transforming a process of the calculus to a set of \emph{species}, and
another for computing the set of possible reactions between species.

An important aspect is that this abstract machine does not have a
native notion of compartment, which is central in the Brane Calculus
(as in any other model of membranes).  To overcome this problem, we
adopt a ``flat'' representation of membrane systems, used also in
\cite{lpp:tcs12}, where the hierarchical structure is represented by
means of \emph{names}: each name represents a compartment, and each
species is labelled with the name of the compartment where it is
located, and the name of its inner compartment (if any). So names and
species are the nodes and the arcs of the tree, respectively.  This
technique can be used for representing any system with a tree-like
structure of compartments.

However, this approach does not scale well, as the population of
species may grow enormously: for instance, a population of $n$
identical cells would lead to $n$ species, all differing only for the
name of its inner compartment, instead of a single specie with
multiplicity $n$.  For circumventing this problem, in
Section~\ref{sec:gsam} we introduce a variant of the GSAM with a
\emph{copy-on-write} optimization strategy---hence called COWGSAM.
The idea is to keep a single copy of each species, with its
multiplicity; when a reaction has to be applied, fresh copies of the
compartments involved are generated on-the-fly, and reactions and
rates are updated accordingly.  In this way, the hierarchical
structure is unfolded only if and when needed.

In Section~\ref{sec:translation} we show how the Brane Calculus can be
represented in the COWGSAM, and we will prove that the abstract
machine obtained in this way is adequate with respect to the
stochastic semantics of the Brane Calculus; in this proof, we take
advantage of the compositional definition of this semantics.

Conclusions and final remarks are in Section~\ref{sec:concl}.

\section{Brane Calculus}
\label{sec:bc}

In this section we recall Cardelli's Brane Calculus
\cite{cardelli04:bc} focusing on its basic version (without
communication primitives, complexes and replication).

First, let us fix the notation we will use hereafter.  Let $S$ be a
set of \emph{sorts} (or ``types''), ranged over by $s,t$, and $T$ a
set of $S$-sorted terms; for $t \in S$, $T_{t} \subseteq T$ denotes
the set of terms of sort $t$.  For $A$ a set of symbols, $A^{*}$
denotes the set of finite words (or lists) over $A$, and
$\tuple{a_{1}, \dots, a_{n}}$ denotes a word in $A^{*}$.  For a word
$\tuple{t_{1}, \dots, t_{n}}$ in $S^{*}$, we define $T_{\tuple{t_{1},
    \dots, t_{n}}} \triangleq T_{t_{1}} \times \dots \times
T_{t_{n}}$.

\paragraph{Syntax}
\label{sec:syntax}
The sorts and the set $\B$ of terms of Brane Calculus are the following:
\begin{align*}
&\text{Sorts} :: S & t \mathrel{::=} {} & \sys \Mid \mem 
\\[0.5ex]
&\text{Membranes} :: \B_\mem  & \sigma,\tau \mathrel{::=} {} 
	& \zero \Mid \sigma|\tau \Mid 
	    \phago_{n}.\sigma \Mid \cophago_{n}(\tau).\sigma \Mid 
           \exo_{n}.\sigma \Mid \coexo_{n}.\sigma \mid \pino_{n}(\tau).\sigma
\\[0.5ex]
&\text{Systems} :: \B_\sys & P,Q \mathrel{::=} {} &
    \void \Mid P \comp Q \Mid \cell{\sigma}{P}
\end{align*}
The subscripted names $n$ are taken from a countable set $\Lambda$.
By convention we shall use $M$, $N$, \dots\ to denote generic Brane 
Calculus terms in $\B$.

A \emph{membrane} can be either the empty membrane $\zero$, or the
parallel composition of two membranes $\sigma | \tau$, or the
action-prefixed membrane $\epsilon.\sigma$.  Actions are:
\emph{phagocytosis} $\phago$, \emph{exocytosis} $\exo$, and
\emph{pinocytosis} $\pino$.  Each action but pinocytosis comes with a
matching co-action, indicated by the superscript $^\bot$.

A \emph{system} can be either the empty system $\void$, or the
parallel composition $P \comp Q$, or the system nested within a
membrane $\cell{\sigma}{P}$.  Notice that, differently
from~\cite{cardelli04:bc}, pino actions are indexed by names in
$\Lambda$.  In~\cite{cardelli04:bc}, names are meant only to pair up
an action with its corresponding co-action, hence a pino action does
not need to be indexed by any name.  Actually, names can be thought of
as an abstract representation of particular protein conformational
shapes; hence, each name can correspond to a different biological
behaviour. Therefore, if we want to observe also kinetic properties of
processes, it is important to keep track of names in pino actions.

Terms can be rearranged according to a structural congruence
relation; the intended meaning is that two congruent
terms actually denote the same system.
Structural congruence $\equiv$ is the smallest equivalence relation 
over $\B$ which satisfies the axioms and rules listed below. 
\begin{gather*}
P\comp Q \equiv Q \comp P \qquad
P\comp(Q\comp R) \equiv (P\comp Q)\comp R \qquad
P\comp \void \equiv P
\\[1ex]
\sigma|\tau \equiv \tau | \sigma \qquad
\sigma |(\tau | \rho) \equiv (\sigma | \tau) | \rho \qquad
\sigma | \zero \equiv \sigma
\\[1ex]
\cell{\zero}{\void} \equiv \void \qquad
\rl{P \comp R \equiv Q \comp R}{P \equiv Q} \qquad
\rl{\sigma | \rho \equiv \tau | \rho}{\sigma \equiv \tau} \qquad
\rl{\cell{\sigma}{P} \equiv \cell{\tau}{Q}}{
  P \equiv Q  \quad \sigma \equiv \tau}
\\[1ex]
\rl{\alpha.\sigma \equiv \alpha.\tau}{
  \alpha \in \set{\phago_{n}, \exo_{n}, \coexo_{n}}_{n \in \Lambda} 
  \quad 
  \sigma \equiv \tau} 
\qquad
\rl{\beta(\rho).\sigma \equiv \beta(\nu).\tau}{
  \beta \in \set{\cophago_{n}, \pino_{n}}_{n \in \Lambda} 
  \quad
  \rho \equiv \nu
  \quad
  \sigma \equiv \tau} 
\end{gather*}
Differently from \cite{cardelli04:bc}, we allow to
rearrange also the sub-membranes contained in co-phago and pino
actions (by means of the last inference rule above).

\paragraph{Reduction Semantics}
\label{sec:redsemantics}

The dynamic behaviour of Brane Calculus is specified by means of a
reduction semantics, defined over a reduction relation (``reaction'')
$\react \subseteq \B_\sys \times \B_\sys$, whose rules are listed in
Table~\ref{tbl:redrules}.
\begin{table}[t]
\hrule
\begin{gather*}
\rl{\cell{\cophago_n(\rho).\tau | \tau_0}{Q} \comp \cell{\phago_n.\sigma | \sigma_0}{P} \react \cell{\tau | \tau_0}{\cell{\rho}{\cell{\sigma | \sigma_0}{P}} \comp Q}}{}
\rlabel{red-phago}
\\[0.5ex]
\rl{\cell{\coexo_n.\tau | \tau_0}{\cell{\exo_n.\sigma | \sigma_0}{P} \comp Q} \react \cell{\sigma | \sigma_0 | \tau | \tau_0}{Q} \comp P}{}
\rlabel{red-exo}
\\[0.5ex]
\rl{\cell{\pino(\rho).\sigma | \sigma_0}{P} \react \cell{\sigma | \sigma_0}{\cell{\rho}{\void} \comp P}}{}
\rlabel{red-pino}
\qquad\qquad
\rl{\cell{\sigma}{P} \react \cell{\sigma}{Q}}{P \react Q}
\rlabel{red-loc}
\\[0.5ex]
\rl{P \comp R \react Q \comp R}{P \react Q}
\rlabel{red-comp}
\qquad\qquad
\rl{P \react Q}{P \equiv P' \quad P' \react Q' \quad Q'\equiv Q}
\rlabel{red-equiv}
\end{gather*}
\hrule
\caption{Reduction semantics for the Brane Calculus.}
\label{tbl:redrules}
\end{table}
Notice that the presence of (red-phago/exo/pino) and (red-equiv) makes
this not a \emph{structural} presentation, since these rules are not
primitive recursive in the syntax (i.e., structural recursive) as
required by the SOS format.

\section{Stochastic Structural Operational Semantics for the Brane Calculus}\label{sec:sgsos}

In this section we recall the stochastic structural operational
semantics for the Brane Calculus, as defined in \cite{bm:tcs12}.
Following \cite{cm:quest10}, we replace the classic ``pointwise''
rules of the form $P \trs{a,r} P'$ with rules of the form $P \to \mu$,
where $\mu$ is an indexed class of \emph{measures} on the
\emph{measurable space} of processes. We assume the reader to be
familiar with basic notions from measure theory; for a brief summary,
see Appendix~\ref{app:measuretheory}.

The set of action labels for the Brane Calculus will be denoted by
$\A$ and can be partitioned with respect to the source sort (i.e.,
either systems or membranes), as follows:
\begin{align*}
 \A_{\sys}  \triangleq {} &
     \{id:\sys\to\tuple{\sys}\} \cup 
     \{\ph_n:\sys\to\tuple{\sys,\sys} \mid  n \in \Lambda\} \cup {} \\
 &  \{\coph_n:\sys\to\tuple{\mem,\mem,\sys,\sys} \mid  n \in \Lambda\} \cup {} \\
 &  \{\ex_n:\sys\to\tuple{\mem,\sys,\sys} \mid  n \in \Lambda\}
 \\
 \A_{\mem} \triangleq {} &
     \set{\phago_n,\exo_n,\coexo_n:\mem\to\tuple{\mem}   \mid n \in \Lambda}\cup{}\\
 &  \set{\cophago_n,\pino_n:\mem\to\tuple{\mem,\mem}  \mid n \in \Lambda} 
\end{align*}
Let $a$ range over $\A$, and $ar(a)$ denote its arity.  To ease the
reading in the following we will use the notation $\Delta_{a}(T,
\Sigma)$ to denote the set of measures $\Delta(T_{\tuple{t_{1}, \dots
    t_{n}}}, \bigotimes_{i=1}^{n} \Sigma_{t_{i}})$, for $ar(a) = t \to
\tuple{t_{1}, \dots, t_{n}}$.

Let $\B/_{\equiv}$ be the set of $\equiv$-equivalence classes on $\B$.
For $M \in \B$, we denote by $[M]_{\equiv}$ the $\equiv$-equivalence
class of $M$ (sometimes dropping the equivalence symbol when clear
from the context).

\begin{definition}[Measurable space of terms]\label{def:termspace}
The \emph{measurable space of terms} $(\B, \Pi)$ is given by 
the measurable space over $\B$ where $\Pi$ is the $\sigma$-algebra 
generated by $\B/_{\equiv}$.
\end{definition}

Notice that $\B/_{\equiv}$ is a denumerable partition of $\B$, hence
it is a base (a generator such that all its elements are disjoint) for
$\Pi$.  Any element of $\Pi$ can be obtained by a countable union of
elements of the base, i.e., for all $\mset{M}\in \Pi$ there exist
$\set{M_i}_{i\in I}$, for some countable $I$, such that $\mset{M} =
\bigcup_{i\in I} [M_i]_\equiv$. As a consequence, in order to generate
the whole $\Pi$ we can simply compute all these unions, without the
need of any closure by complement.

A similar argument holds for the product space $(\B_{\tuple{t_{1},
    \dots, t_{n}}}, \bigotimes_{i=1}^{n} \Pi_{t_{i}})$, where $t_{i}
\in \set{\mem, \sys}$ ($1 \leq i \leq n$); indeed
$\bigotimes_{i=1}^{n} \Pi_{t_{i}}$ can be generated from the base
$\B_{\tuple{t_{1}, \dots, t_{n}}}/_{\equiv_{\tuple{t_{1}, \dots,
      t_{n}}}}$, where $\equiv_{\tuple{t_{1}, \dots, t_{n}}} \subseteq
\B_{\tuple{t_{1}, \dots, t_{n}}} \times \B_{\tuple{t_{1}, \dots,
    t_{n}}}$ is defined by
\begin{align*}
  \tuple{M_{1}, \dots, M_{n}} 
  &\equiv_{\tuple{t_{1}, \dots, t_{n}}} 
  \tuple{N_{1}, \dots, N_{n}}
  &\text{iff}&&
  M_{i} &\equiv N_{i}, \text{ for all $1 \leq i \leq n$} \,, 
\end{align*}
which can be easily checked to be an equivalence relation.
$\equiv_{\tuple{t_{1}, \dots, t_{n}}}$-equivalence classes are
rectangles, i.e.\ $[\tuple{M_{1}, \dots,
  M_{n}}]_{\equiv_{\tuple{t_{1}, \dots, t_{n}}}} = [M_{1}]_{\equiv}
\times \dots \times [M_{n}]_{\equiv}$, therefore the product measure
$\bigotimes_{i=1}^{n} \Pi_{t_{i}}$ is well defined.  For sake of
simplicity in the following we write $[\tuple{M_{1}, \dots,
  M_{n}}]_{\equiv}$ in place of $[\tuple{M_{1}, \dots,
  M_{n}}]_{\equiv_{\tuple{t_{1}, \dots, t_{n}}}}$, and
$\B_{\tuple{t_{1}, \dots, t_{n}}}/_{\equiv}$ in place of
$\B_{\tuple{t_{1}, \dots, t_{n}}}/_{\equiv_{\tuple{t_{1}, \dots,
      t_{n}}}}$.

\begin{table}[t]
\hrule
\begin{gather*}
\rl{\zero \to_\mem \omega_{\mem}}{}
\rlabel{zero}
\qquad
\rl{\alpha.\sigma \to_\mem [\alpha]_{\sigma}}{
  \alpha \in \set{\phago_n, \exo_n, \coexo_n \mid n \in \Lambda}}
\rlabel{pref}
\\[2ex]
\rl{\beta(\tau).\sigma \to_\mem [\beta]_{\sigma}^{\tau}}{
  \beta \in \set{\cophago_n, \pino_n \mid n \in \Lambda}}
\rlabel{pref-arg}
\qquad
\rl{\sigma | \tau \to_\mem \mu' \mpar{\sigma}{\tau} \mu''}{
\sigma \to_\mem \mu' \qquad \tau \to_\mem \mu''}
\rlabel{par}
\\[2ex]
\rl{ \void \to_\sys \omega_{\sys}}{}
\rlabel{void}
\qquad
\rl{ \cell{\sigma}{P} \to_\sys \mu \mAt{\sigma}{P} \nu}{
  \sigma \to_\mem \nu \qquad P \to_\sys \mu}
\rlabel{loc}
\qquad
\rl{ P \comp Q \to_\sys \mu' \mtimes{P}{Q} \mu''}{
  P \to_\sys \mu' \qquad Q \to_\sys \mu''}
\rlabel{comp}
\end{gather*}
\hrule
\caption{Stochastic structural operational semantics for Brane Calculus}
\label{tbl:stochSOS}
\end{table}

The operational semantics associates with each membrane a family of
measures in $\Delta^{\A_\mem}(\B, \Pi)$, and with each system a family
of measures in $\Delta^{\A_\sys}(\B, \Pi)$.  This can be represented
by two relations $\to_\mem: T_\mem \to \Delta^{\A_\mem}(\B, \Pi)$,
$\to_\sys: T_\sys \to \Delta^{\A_\sys}(\B, \Pi)$, defined by the SOS
rules listed in Table~\ref{tbl:stochSOS}. (In the following, for sake
of readability, we will drop the indexes $_\mem, _\sys$).  In these
rules we use some constants and operations over indexed families of
measures, that we define next.  For a set (of labels) $A$, let us
denote by $\Delta^A(\B,\Pi)$ the set $\prod_{a \in A} \Delta_{a}(\B,
\Pi)$ of $A$-indexed families of measures over $(\B,\Pi)$.  Given a
family of measures $\mu \in \Delta^A(\B,\Pi)$ and $a\in A$, the
$a$-component of $\mu$ will be denoted as $\mu_a \in \Delta_{a}(\B,
\Pi)$.

\begin{description}[leftmargin=0pt]

\item[Null:] Let $\omega_{\mem} \in \Delta^{\A_\mem}(\B, \Pi)$ be the
  constantly zero measure, i.e., for all $a {\in} \A_{\mem}$ such that
  $ar(a) = t \to \tuple{t_{1}, \dots, t_{n}}$ and $\mset{M} \in
  \bigotimes_{i=1}^{n} \Pi_{t_{i}}$: $(\omega_{\mem})_a(\mset{M}) {=}
  0$.

\item[Prefix:] For arbitrary $n \in \Lambda$, 
$\alpha \in \set{\phago, \exo, \coexo}$, and
$\beta \in \set{\cophago, \pino}$, let the constants
$[\alpha_{n}]_{\sigma},[\beta_{n}]_{\sigma}^{\tau} \in
\Delta^{\A_\mem}(\B, \Pi)$ be defined, for arbitrary $X,Y \in
\B_{\mem}/_{\equiv}$, as follows:
\begin{align*}
  ([\alpha_{n}]_{\sigma})_{\alpha_{m}}(X) & = 
    \begin{cases}
      \iota(n)		&\text{if } n = m \text{ and } \sigma \in X \\
      0			&\text{otherwise}
    \end{cases}
  \\
  ([\alpha_{n}]_{\sigma})_{\beta_{m}}(X \times Y) &= 0 
 \\[1ex]
   ([\beta_{n}]_{\sigma})_{\alpha_{m}}(X) &= 0
  \\
  ([\beta_{n}]_{\sigma}^{\tau})_{\beta_{m}}(X \times Y) & = 
    \begin{cases}
      \iota(n)		&\text{if $n = m$ and $\sigma \in X$, $\tau \in Y$} \\
      0			&\text{otherwise}
    \end{cases}
\end{align*}

\item[Parallel:] For $\mu, \mu' \in \Delta^{\A_\mem}(\B, \Pi)$, let
  $\mu \mpar{\sigma}{\tau} \mu' \in \Delta^{\A_\mem}(\B, \Pi)$ be
  defined, for $n \in \Lambda$, $\alpha \in \set{\phago, \exo,
    \coexo}$, $\beta \in \set{\cophago, \pino}$, and $X,Y \in
  \B_{\mem}/_{\equiv}$, as follows (where for all $X,\tau$: $X_{|\tau}
  \triangleq \bigcup \set{[\sigma]_{\equiv} \mid \sigma|\tau \in X}$):
\begin{align*}
    (\mu \mpar{\sigma}{\tau} \mu')_{\alpha_{n}}(X) & =
       \mu_{\alpha_{n}}(X_{| \tau}) + \mu'_{\alpha_{n}}(X_{| \sigma})
    \\[1ex]
    (\mu \mpar{\sigma}{\tau} \mu')_{\beta_{n}}(X \times Y) &=
       (\mu)_{\beta_n}(X_{| \tau} \times Y) + (\mu')_{\beta_n}(X_{| \sigma} \times Y)
  \end{align*}

\item[Void:] Let $\omega_{\sys} \in  \Delta^{\A_\sys}(\B, \Pi)$ 
be defined by $(\omega_{\sys})_a(\mset{M}) = 0$ for any $a \in \A_{\sys}$, 
such that $ar(a) = t \to \tuple{t_{1}, \dots, t_{n}}$, and 
$\mset{M} \in \bigotimes_{i=1}^{n} \Pi_{t_{i}}$.

\item[Nesting:] For 
$\nu \in \Delta^{\A_\mem}(\B, \Pi)$ and 
$\mu \in \Delta^{\A_\sys}(\B, \Pi)$, let
$\mu \mAt{\sigma}{P} \nu \in \Delta^{\A_\sys}(\B, \Pi)$ be
defined, for $X, Y \in \B_{\mem}/_{\equiv}$ and $Z, W \in
\B_{\sys}/_{\equiv}$, as follows:
\begin{align*}
  (\mu \mAt{\sigma}{P} \nu)_{\ph_{n}}(Z \times W) &= 
  \begin{cases}
    \nu_{\phago_{n}}([\sigma]_{\equiv})	
    	&\text{if } \cell{\sigma}{P} \in Z \text{ and } \void \in W \\
    0	&\text{otherwise}
  \end{cases}
  \\
  (\mu \mAt{\sigma}{P} \nu)_{\coph_{n}}(X \times Y \times Z \times W) &=
  \begin{cases}
  \nu_{\cophago_{n}}(X \times Y) 
  	&\text{if } P\in Z \text{ and } \void \in W \\
  0	&\text{otherwise}
  \end{cases}
\\
(\mu \mAt{\sigma}{P} \nu)_{\ex_{n}}(X \times Z \times W) &=
  \begin{cases}
    \nu_{\exo_{n}}(X)	&\text{if } P \in Z \text{ and } \void \in W \\
    0	&\text{otherwise}
  \end{cases}
\end{align*}
\begin{align*}
  (\mu \mAt{\sigma}{P} \nu)_{id}(X) &= 
    \mu_{id}(X_{ \cell{\sigma}{} }) + 
    \sum^{n \in \Lambda}_{
    \cell{X'}{ \cell{X''}{[\void]_{\equiv}} \comp [P]_{\equiv} } = X}
    \nu_{\pino_{n}}(X' \times X'') + {}
     \\
   &\hspace{1em}
    \sum^{n \in \Lambda}_{\cell{X' | X'' }{Y''} \comp Y' = X}
    \hspace{-1em}
    \frac{
    \mu_{\ex_{n}}(X' \times Y' \times Y'') \cdot
    \nu_{\coexo_{n}}(X'')}{\iota(n)}
\end{align*}

\item[Composition:]  For 
$\mu, \mu' \in \Delta^{\A_\sys}(\B, \Pi)$, let
$\mu \mtimes{P}{Q} \mu' \in \Delta^{\A_\sys}(\B, \Pi)$ be
defined, for $X, Y \in \B_{\mem}/_{\equiv}$ and $Z, W \in
\B_{\sys}/_{\equiv}$, as follows (where for all $W,Q$, $W_{\comp Q} \triangleq
  \bigcup \set{[P]_{\equiv} \mid P \comp Q \in W}$):
\begin{align*}
  (\mu \mtimes{P}{Q} \mu')_{\ph_{n}}(Z \times W) &= 
    \mu_{\ph_{n}}(Z \times W_{\comp Q}) + 
    \mu'_{\ph_{n}}(Z \times W_{\comp P})
  \\[1ex]
  (\mu \mtimes{P}{Q} \mu')_{\coph_{n}}(X \times Y \times Z \times W) &= 
    \mu_{\coph_{n}}(X \times  Y \times Z \times W_{\comp Q}) + 
    \mu'_{\coph_{n}}(X \times Y \times  Z \times W_{\comp P})
  \\[1ex]
  (\mu \mtimes{P}{Q} \mu')_{\ex_{n}}(X \times Z \times W) & = 
    \mu_{\ex_{n}}(X \times Z \times W_{\comp Q}) + 
    \mu'_{\ex_{n}}(X \times Z \times W_{\comp P})
  \\[1ex]
  (\mu \mtimes{P}{Q} \mu')_{id}(X) & = 
    \mu_{id}(X_{\comp Q}) + \mu'_{id}(X_{\comp P}) + {}
    \\
  &
    \sum^{n \in \Lambda}_{
      \cell{X_{1}}{ \cell{X_{2}}{ Y_{1} } \comp Z_{1} } \comp Y_{2} \comp Z_{2} = X}
    \hspace{-2em}
    \frac{
      \mu_{\ph_{n}}( Y_{1} \times Y_{2})
      \cdot
      \mu'_{\coph_{n}}( X_{1} \times X_{2} \times Z_{1} \times Z_{2} )
    }{\iota(n)} + {}
    \\
  &
    \sum^{n \in \Lambda}_{
      \cell{X_{1}}{ \cell{X_{2}}{ Z_{1} } \comp Y_{1} } \comp Z_{2} \comp Y_{2} = X}
    \hspace{-2em}
    \frac{
      \mu_{\coph_{n}}( X_{1} \times X_{2} \times Y_{1} \times Y_{2} )
      \cdot
      \mu'_{\ph_{n}}( Z_{1} \times Z_{2} )
    }{\iota(n)}
\end{align*}
\end{description}

These operators have nice algebraic properties (e.g., $\mu'
\mpar{\sigma}{\tau} \mu'' = \mu'' \mpar{\tau}{\sigma} \mu'$, $(\mu'
\mpar{\sigma}{\tau} \mu'') \mpar{\sigma|\tau}{\rho} \mu''' = \mu'
\mpar{\sigma}{\tau|\rho} (\mu'' \mpar{\tau}{\rho} \mu''')$, \dots),
and respect the structural congruence (e.g., if $P \equiv P'$ and $Q
\equiv Q'$ then $\mu' \mtimes{P}{Q} \mu'' = \mu' \mtimes{P'}{Q'}
\mu''$).  We refer to \cite{bm:tcs12} for further details about these
properties, very useful in calculations.

The next lemmata state that the stochastic transition relation $\to$
(and hence the operational semantics) is well-defined and consistent,
that is, for each process we have exactly one family of measures of
its continuations, and this family respects structural congruence.
\begin{lemma}[Uniqueness] \label{lm:uniqueness} 
 For $a \in \A$ such that $ar(a) = t \to \tuple{t_{1}, \dots, t_{n}}$, and  
 $M \in \B_{t}$, there exists a unique 
 $\mu \in \Delta^{\A_t}(\B, \Pi)$ such that $M \to \mu$.
\end{lemma}
\begin{lemma} \label{lm:stransitionisuptoequiv}
 If $M \equiv N$ and $M \to \mu$, then $N \to \mu$.
\end{lemma}

This operational semantics can be used to define the 
`` traditional'' pointwise semantics:
\begin{equation*}
  M \trs{a, r} \tuple{M_{1}, \dots, M_{n}}
  \stackrel{\triangle}{\iff}
  M \to \mu \; \text{ and } \; \mu_a([\tuple{M_{1}, \dots, M_{n}}]_{\equiv}) = r
\end{equation*}
and it is conservative with respect to the non-stochastic reduction semantics.
\begin{proposition} \label{prop:conservative}
  For all systems $P,Q\in\B_\mem$, 
   if $P \to \mu$ and $\mu_{id}([Q])>0$ then $P \react Q$.
\end{proposition}

\section{The COW Generic Stochastic Abstract Machine}
\label{sec:gsam}
In this section we present a variant of the generic stochastic
abstract machine (GSAM), oriented to systems with nested
compartments.

The GSAM has been introduced in \cite{pylp:cmsb10,lpp:tcs12} for
simulating a broad range of process calculi with an arbitrary
reaction-based simulation algorithm.  Although it does not have a
native notion of ``compartment'', nested systems can represented by
``flattening'' all compartments and species into a single multiset,
where each species is tagged with names representing their position in
the hierarchy, as shown in \cite{lpp:tcs12}.  The idea is to represent
each compartment as a different species, keeping track of their
position in the hierarchy by means of \emph{(node) names}. These names
are ranged over by $x,y,z,\dots$, and are different from names in
actions.  As an example, a system $\cell{\sigma}{\cell{\tau}{}}$ is
represented as the multiset $\{\cell{\sigma}{^x_y}\mapsto
1,\cell{\tau}{^y_z}\mapsto 1\}$, which means ``there is one cell with
membrane $\sigma$ located in the compartment $x$ and whose compartment
is $y$, and one cell with membrane $\tau$ positioned in $y$ and whose
compartment is $z$''.  Reactions can happen only if the names tagging
the involved species match according to the required nesting
structure.

Unfortunately, this approach does not scale well, as the population of
species grows.  Let us consider a system composed of $n$ copies of the
same cell, e.g., $n\cdot(\cell{\sigma}{})$ (where $n$ can be easily in
the order of $10^3$--$10^4$).  In the original GSAM idea, this should
be represented in the machine as a single species with multiplicity
$n$, and each possible reaction is represented once but with
propensity given by the law of mass action taking into account the
species' multiplicity $n$.  Instead, the ``flat encoding'' above
yields $n$ different species $\cell{\sigma}{^x_{y_1}}, \dots,
\cell{\sigma}{^x_{y_n}}$, each with multiplicity 1; the set of
reactions explodes correspondingly.

For circumventing this problem we introduce a variant of the GSAM
with a \emph{copy-on-write} strategy---hence it is called COWGSAM.
The idea is to keep a single copy of each species, with its
multiplicity; the same applies to reactions.  When a reaction has
to be applied, the compartments involved are ``unfolded'', i.e., fresh
copies of the compartments are generated and the reaction set is
modified accordingly; then, the reaction can be applied.  In this way,
the hierarchical structure is unfolded only if and when needed.

In order to implement this idea, we have to modify the notion of
machine term, reaction and reaction rule.  The COWGSAM (with the Next
Reaction method) is shown in Figure~\ref{fig:gsam}.

\begin{figure}[t]
\begin{align*}
T ::= & E\vdash (t,S,R)   \tag{Machine term}\\
E ::= & x_1, \ldots, x_n \tag{Environments}\\
S ::= &\{I_1 \mapsto i_1, \ldots, I_N \mapsto i_N\}		\tag{Populations}\\
R ::= & \{O_1 \mapsto A_1, \ldots, O_N \mapsto A_N\}         \tag{Reactions}\\
O ::= & (S_1,r,f,S_2)  \tag{Reaction}
\end{align*}
\vspace{-1ex}
\begin{equation}
\rl{E \vdash (t, S, R) \trs{a, (S_1,r,f,S_2)} norm(E' \cup fn(S_2) \vdash (f(S_2 \oplus ((t',S',R') \ominus
  S_1))))}%
  {((S_1,r,f,S_2), a, t') = next(t,S,R) \quad E'\vdash(t,S',R') = cow(E \vdash (t,S,R),S_1)}
\tag{Reaction rule}\label{rl:reac}
\end{equation}
\vspace{-1ex}
\begin{align*}
next(t,S,R) \triangleq &\ (O,a,t') \quad \text{if } R(O) = (t',a) \text{ and } t'
= \min\{t \mid R(O) = (t,a)\}
\\[1ex]
cow(E \vdash (t,S,R),\emptyset) \triangleq &\ E \vdash (t,S,R)
\\
cow(E \vdash (t,S,R),\{\cell{\rho}{^y_z}\mapsto j\}\cup S_1 ) \triangleq &\
cow(E'\cup fn(S'') \vdash (t, S'', R'' \cup init(L'\cup L'',(t,S',R)), S_1) \quad\text{where} \\
& \text{if } S(\cell{\sigma}{^x_y}) > 1 \text{ for some } \sigma, x: \\
& \quad \text{let } i = S(\cell{\sigma}{^x_y})  \text{ and } y' \notin E\\
& \quad (S',R') = dup(E\vdash (t,S\{\cell{\sigma}{^x_y} \mapsto 1,\cell{\sigma}{^x_{y'}} \mapsto i - 1\},R),y,y')\\
& \quad L' = reactions(\cell{\sigma}{^x_{y'}}\mapsto i-1, S)\\
& \text{otherwise let } (S',R') = (S,R),  L' = \emptyset \text{ in}\\
& \text{let } n = S(\cell{\rho}{^y_z}),  E'=E\cup fn(S'), \text{ and } z' \notin E'\\
& \quad (S'',R'') = dup(E' \vdash (t, S'\{\cell{\rho}{^y_z}\mapsto j,\cell{\rho}{^y_{z'}} \mapsto n{-}j\},R'),z,z')\\
& \quad L'' = reactions(\cell{\rho}{^y_{z'}}\mapsto n-j, S')
\\[1ex]
dup(E \vdash (t,S,R),y,y') \triangleq &\ (S' \cup S'', R\cup R'' \cup init(L,(t,S',R)))\\
&\text{where } S' = S \cup \{\cell{\rho}{^{y'} _{w'}} \mapsto i \mid S(\cell{\sigma}{^z _y}) > 0, S(\cell{\rho}{^y_w}) = i\}, w' \notin E,\\
& \phantom{where } L = reactions(\cell{\rho}{^{y'}_{w'}}\mapsto i, S'),  \\
& \phantom{where } (S'',R'') = dup(E \cup \{w'\} \vdash (t,S,R),w,w') 
\\[-5ex]
\end{align*}
\caption{The COW Generic Stochastic Abstract Machine, with the Next
  Reaction method.}\label{fig:gsam}
\vspace{-1ex}
\end{figure}

\begin{figure}[t]\ContinuedFloat
\begin{align*}
\{I_1 \mapsto i_1, \ldots, I_N \mapsto i_N\} \oplus (t,S,R) \triangleq&\  I_1 \mapsto i_1 \oplus \ldots \oplus I_N \mapsto i_N \oplus (t,S,R)\\
I \mapsto i \oplus (t,S,R)	\triangleq &
 \begin{cases}
  (t, S', R \cup updates(I,(t,S',R))) &  \text{if } S(I) = i' \text{
    and } S' = S\{I \mapsto i' + i\} \\
  (t, S', R \cup init(L,(t,S',R))) &  \text{if }  I \notin dom(S), 
    S' = S\{I \mapsto i\} \\& \text{and } L = reactions(I\mapsto i,S)
  \end{cases} 				
\\[1ex]
(t,S,R) \ominus \{I_1 \mapsto i_1, \ldots, I_N \mapsto i_N\} \triangleq&\ (t,S,R) \ominus I_1 \mapsto i_1 \ominus \ldots \ominus I_N \mapsto i_N
\\
 (t,S,R) \ominus I \mapsto i \triangleq&\ (t, S', R \cup updates(I,(t,S',R)))\\
 & \hspace{2cm} \text{if }  S(I) = i', i' \geq i \text{ and } S' = S\{I \mapsto i'-i\}
\\[1ex]
init(L,(t,S,R)) \triangleq &\ \{O \mapsto (t',a) \mid O \in L \text{ and } a =
propensity(O,S) \text{ and } \\
&\hspace{22.5mm} O = (S_1,r,f,S_2) \text{ and } t' = t + delay(r,a)\}
\\[1ex]
updates(I, (t,S,R)) \triangleq &\ \{O \mapsto (t',a') \mid R(O) =
(t'',a)  \text{ and } O = (S_1,r,f,S_2) \text{ and } S_1(I) > 0 \text{ and}\\
&\hspace{11mm} \text{if } t'' > t  \text{ then } a' = propensity(O,S) \text{ and } t' = t + (a/a')(t''-t)\\
&\hspace{11mm}  \text{if } t'' = t \text{ then } a' = propensity(O,S) \text{ and } t' = t + delay(r,a)\}
\\[1ex]
propensity((S_1,r,f,S_2), S) \triangleq &\ r \cdot \left( \begin{array}{c} S^*(I_1) \\ j_1 \end{array} \right) \cdot \ldots \cdot \left( \begin{array}{c} S^*(I_n) \\ j_n \end{array} \right) \text{ if } S_1 = \{I_1 \mapsto j_1, \ldots, I_n \mapsto j_n\}
\\[1ex]
S^*(\cell{\sigma}{^x_y}) \triangleq &
\begin{cases}
S(\cell{\sigma}{^x_y}) &\text{ if } x = root \\
S(\cell{\sigma}{^x_y}) \cdot S^*(\cell{\rho}{^z_x}) &\text{ if } x \neq root \text{ and } S(\cell{\rho}{^z_x}) > 0
\end{cases}
\end{align*}
\vspace{-2ex}
\caption{The COW Generic Stochastic Abstract Machine (cont.).}
\vspace{-1ex}
\end{figure}

First, for generating fresh names, we have to keep track of those
already allocated. To this end we introduce \emph{environments}, which
are finite sets of names.  Then, the machine state is represented by a
\emph{machine term} $T$, i.e.~a quadruple $E\vdash (t, S, R)$ where
$E$ is an environment; $t$ is the current time; $S$ is a finite
function mapping each species $I$ to its population $S(I)$ (if not
null); and $R$ maps each reaction $O$ to its activity $A$, which is
used to compute the next reaction.  (Notice that the syntax of species
$I$ is left unspecified, as it depends on the specific process
calculus one has to implement.)  We say that a machine term $E \vdash
(t,S,R)$ is \emph{well-formed} if for all $x_i,x_j \in E: x_i=x_j
\Rightarrow i=j$, and the free names in $S,R$ appear in $E$.  In the
following, we assume that machine terms are well formed.

Each reaction is a quadruple $(S_1,r,f,S_2)$, basically representing a
reaction $S_1 \trs{r} S_2$, where
\begin{itemize}
\item $S_1$ and $S_2$ denote the \emph{reactant} population and
  \emph{product} population respectively;
\item \textit{r} denotes the rate (in $s^{-1}$) of the reaction;
\item \textit{f} is a function mapping machine terms to machine terms;
  this functions implements any global update of the machine term
  after the reaction (if needed).
\end{itemize}

The \emph{transitions} of the abstract machine are represented by a
relation $T \trs{a,O} T'$ between machine terms, indexed by the
propensity $a$ and reaction rules. This should be read as ``$T$ goes
to $T'$ with rate $a$, using the rule $O$''.  This relation is defined
by (\ref{rl:reac}) in Figure~\ref{fig:gsam}, where the function
$next(T)$ selects the next reaction, i.e.~it returns a pair $(O,t')$
where $O=(S_1,r,f,S_2)$ is the reaction to happen first among all
possible reactions in $T$, and $t'$ is the new time of the system.
Once the reaction has been selected, we have first to create the
separate (private) copies of the compartments involved, and to update
the reaction set accordingly. This is done by the functions $cow(\_)$
and $dup(\_)$, which implement a deep copy-on-write:
$cow(E\vdash(t,S,R), S_1)$ is a machine term $E'\vdash(t,S',R')$
representing the same state as $E\vdash(t,S,R)$, but in $S'$ the
species indicated in $S_1$ are unfolded; $E'$ contains all names which
have been generated in the process, and $R'$ is the new set of
reactions.  (Actually, the freshly generated copies represent the
instances which are \emph{not} involved by the reaction; this
simplifies the reaction application.) An example of the action of
$cow(\_)$ is depicted in Figure~\ref{fig:cow}.
\begin{figure}[t]
  \centering
  \includegraphics[width=0.7\textwidth]{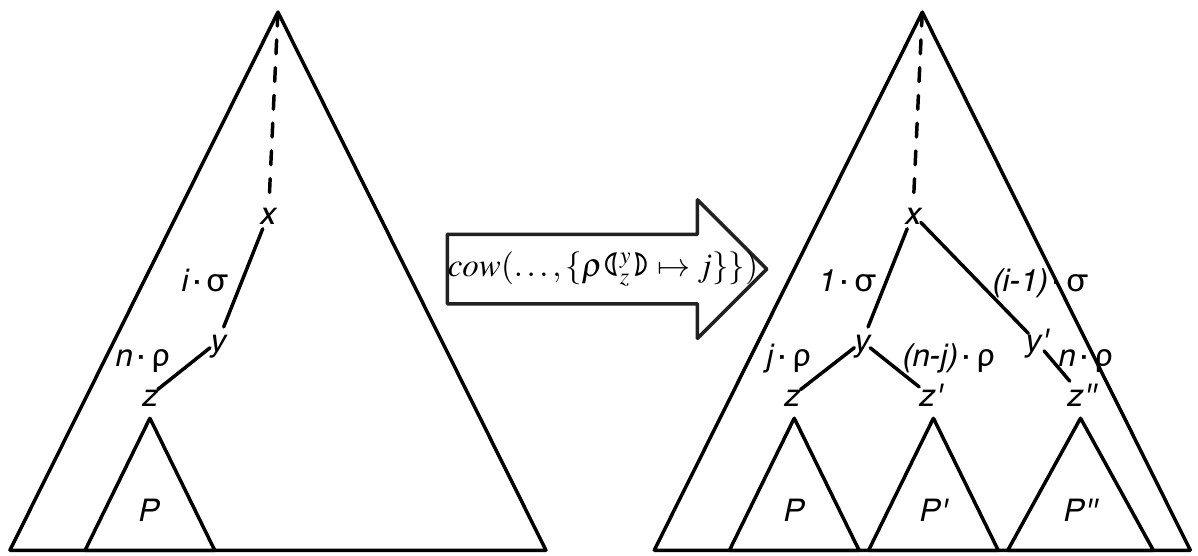}
  \caption{Example copy-on-write of a subterm.}
  \label{fig:cow}
\end{figure}

At this point, the reaction is executed, by removing the reactants
$S_1$ from the machine term (via the operation $\ominus$), adding the
products $S_2$ (via the operation $\oplus$) and updating the current
time of the machine. The function $f$ performs any global ``clean-up''
and restructuring to the machine term that may be required by the
reaction (e.g., garbage collection, elimination of names not used
anymore,\dots).  Moreover, since a reaction can rearrange the
hierarchy structure, possibly creating new compartments and deleting
others, we have to add to the environment any fresh name introduced in
the products.

Finally, the term can be ``normalized'' by collapsing equivalent
copies of the same subtree into a single copy (with the right
multiplicity), by the function $norm(\_)$. In its simplest form,
$norm(\_)$ can be the identity, i.e., no normalization is performed at
all.  Although this is correct, it can lead to unnecessary copies of
the same subtrees.  We can define $norm(\_)$ such that
\[
  norm(\{I_1\mapsto i_1, I_2\mapsto i_2\} \cup S, R) =
  norm(\{I_1\mapsto i_1 + i_2\} \cup S, R[I_2/I_1])
 \quad \text{if }  I_1\equiv I_2
\]
where the equivalence between species can be implemented by comparing
the subtrees starting from $I_1$, $I_2$, e.g., by calculating a
suitable hash value.  We leave this refinement as future development.

In order to implement the Next Reaction method, each reaction $O$ is
associated with a pair $R(O) = (a, t)$, where $a$ is the reaction
\emph{propensity} and $t$ is the time at which the reaction is
supposed to occur. The function $delay(r,a)$ computes a time interval
from a random variable with rate $r$ and
propensity $a$. 

This general structure can be instantiated with a given process
calculus, just by providing the definition for the missing parts.
Given a set $Proc$ of processes, we have to define:
\begin{enumerate}
\item the syntax of the species $I$ (which may be different from that
  of processes);
\item a function $species_{E,x}(P)$ mapping a process $P\in Proc$ to a
  species set located in $x$;
\item a function $reactions(I\mapsto i, S)$ for computing the multiset
  of reactions between a (new) species \textit{I} with multiplicity
  $i$ and a population of (existing) species $S$.
\end{enumerate}
The function $species$ is used to initialise the abstract machine at
the beginning of a simulation.  If we aim to simulate the execution of
a process $P\in Proc$, the corresponding initial state (rooted in $root$)
is
\[
\bamc{P}_{root} \triangleq fn(J) \vdash J \oplus
 (0,\emptyset,\emptyset)
 \quad \text{ where } J = species_{\emptyset,root}(P).
\]

The $reactions$ function is used to adjust the set of possible
reactions dynamically.

\section{Implementing the Stochastic Brane Calculus in
  COWGSAM}\label{sec:translation}
In this section, we show how the COW Generic Stochastic Abstract
Machine can be used to implement the Stochastic Brane Calculus,
following the protocol described in the previous section.

\subsection{Encoding of the Stochastic Brane Calculus}

\paragraph{Syntax of species}
We define the species for the brane calculus, which in turn lead us to
introduce \emph{complexes} and \emph{actions}.  Notice that (despite
the deceiving syntax) species are not systems and complexes are not
membranes; nevertheless, actions are a subset of membranes.
\begin{align*}
I ::=\ & \cell{C}{^x_y}    \tag{Species}\\
C ::=\ & A_1, \ldots, A_n    \tag{Complexes}\\
A ::=\ & \phago_n.\sigma \mid \cophago_n (\tau) .\sigma \mid \exo_n.\sigma \mid \coexo_n.\sigma \mid \pino_n(\tau).\sigma  \tag{Actions}
\end{align*}

Node names can appear in the species in $S$ and in reactions $R$. 

\paragraph{The $species$ function}
We can now provide the definition of the translation of a process
$P\in \B_\sys$ into a species set.  Basically, each compartment is
assigned a different, fresh node name; therefore, the function
$species_{E,x}(P)$ is parametric in the set $E$ of allocated node
names, and the name $x\in E$ to be used as the location of the system
$P$. The name $x$ changes as we descend the compartment hierarchy.

In order to capture correctly the multiplicity of each species, we
assume that systems are in \emph{normal form}. Basically, this form is
a shorthand for products where $n$ copies of the same system, i.e.,
$P\circ\dots\circ P$, are represented as $n\cdot P$.
\begin{multline*}
\text{Normal Systems} :: \B_\sys^n \qquad
Q \mathrel{::=}  
    n_1 \cdot \cell{\sigma_1}{Q_1} \comp \dots \comp n_k \cdot
    \cell{\sigma_k}{Q_k} \\
   \text{ where } k \geq 0 \text{ and for } i \neq j: \cell{\sigma_i}{Q_i} \not\equiv \cell{\sigma_j}{Q_j}
\end{multline*}
A system in normal form can be translated into a system just by
unfolding the products.  For $Q\in\B_\sys^n$ a system in normal form,
let $\lceil Q \rceil \in \B_\sys$ defined as follows:
\[
  \lceil \void \rceil = \void
  \qquad
  \lceil n\cdot \cell{\sigma}{Q'} \circ Q'' \rceil =
  \overbrace{\cell{\sigma}{\lceil Q' \rceil}
    \circ\dots\circ \cell{\sigma}{\lceil Q' \rceil}}^{n \text{ times}} \circ  \lceil Q'' \rceil 
\]

\begin{proposition}
  For all $P\in \B_\sys$, there exists a system in normal form
  $Q\in \B_\sys^n$ such that $\lceil Q \rceil \equiv P$.
\end{proposition}
As a consequence, we can give the definition of $species$ on systems
in normal form, as follows:
\begin{align*}
species_{E,x}(\void) \triangleq\ &\emptyset
\\
species_{E,x}(n \cdot \cell{\sigma}{Q_1} \circ Q_2)  \triangleq\ &
    \{s(\sigma) \lcell^x_y\rcell \mapsto n\} \cup
    species_{E\uplus\{y\},y}(Q_1) \cup species_{E,x}(Q_2) \\
& \text{with }\ y \notin E \text{ and } fn(species_{E\uplus\{y\},y}(Q_1))\cap
fn(species_{E,x}(Q_2)) \subseteq \{x\}
\end{align*}
The condition in the second case ensures that two different
compartments are never given the same name---any name clash has to be
resolved by an $\alpha$-conversion.  The function $s(\_)$ converts a
membrane into a set of complexes:
\begin{alignat*}{3}
s(\zero) \triangleq\ &\emptyset
&\qquad
s(\sigma|\sigma')  \triangleq\ & s(\sigma) \cup s(\sigma')
&\qquad
s(\pi.\sigma)  \triangleq\ & \{\pi. \sigma\}
\end{alignat*}

\paragraph{The $reactions$ function}
The next step is to define the function $reactions(I\mapsto i, S)$,
for $I$ a species with multiplicity $i$ and $S$ a population.
\begin{align*}
reactions_E(I_1\mapsto i, S) \triangleq \	& unary_E(I_1) \cup binary_E(I_1, S)
\\
unary_E(I_1) \triangleq\ & \{(\{I_1 \mapsto 1\},r_n, id,
\{\cell{(s(\sigma) \cup U_1')}{^x_y} \mapsto 1,
   \cell{s(\rho)}{^y_w}\mapsto 1\}) \mid \\
   &  \hspace{2cm} I_1 = \cell{U_1}{^x_y},~ U_1 = \{\pino_n(\rho).\sigma\} \cup U_1', w \notin E \}
\\
binary_E(I_1,S) \triangleq\ & \left\{(\{I_1 \mapsto 1,I_2 \mapsto 1\},r_n,
f, \{\cell{(s(\tau) \cup U'_1 \cup s(\sigma) \cup U'_2)}{^x_y}\mapsto
1\}) \mid  I_2 \in dom(S),\right.\\
& \hspace{2cm}
(I_1 = \cell{U_1}{^x_y}, I_2 = \cell{U_2}{^y_w})
\vee
(I_2 = \cell{U_1}{^x_y}, I_1 = \cell{U_2}{^y_w}), \\
& \hspace{2cm} \left.
U_1 = \{\coexo_n.\tau\} \cup U'_1,
U_2 = \{\exo_n.\sigma\} \cup U'_2,
f = \lambda T.T[w :=x] \right\} \ \uplus
\\
& \left\{(\{I_1\mapsto 1,I_2\mapsto 1\},r_n, id,\phantom{U'_1}\right.\\
&\hspace{1cm} \{\cell{(s(\tau) \cup
  U'_1)}{^x_y} \mapsto 1, \cell{s(\rho)}{^y_w} \mapsto 1, \cell{(s(\sigma) \cup
  U'_2)}{^w_z}\mapsto 1 \}) \mid I_2\in dom(S), \\
&\hspace{2cm}
(I_1 = \cell{U_1}{^x_y}, I_2 = \cell{U_2}{^x_z})
\vee
(I_2 = \cell{U_1}{^x_y}, I_1 = \cell{U_2}{^x_z}),
\\
& \hspace{2cm}\left.
U_1 = \{\cophago_n(\rho).\tau\} \cup U'_1,
U_2 = \{\phago_n.\sigma\} \cup U'_2,~ w \notin E \right\}
\end{align*}
In the case of Brane Calculus, the unary reactions are only those
arising from pinocytosis, while binary reactions arise from exocytosis
and phagocytosis. In both cases, the multiplicity of each reactant is
1, so the multiplicity of $I_1$ is not relevant.  Exocytosis merges
two compartments; this is reflected by the fact that the
``rearranging'' function $f$ substitutes every occurrence of the name
$w$ in $T$ with $x$.  On the other hand, pinocytosis and phagocytosis
create new compartments; to represent the new structure, we choose a
fresh name $w$ representing the new intermediate nesting level, and
reconnect the various compartments accordingly.  Therefore, for any
reaction $(S_1,r,f,S_2) \in reactions_E(I_1 \mapsto i, S)$, $fn(S_2)
\setminus E$ is either $\emptyset$ (in the case of exocytosis)
or the singleton $\{ w \}$.

\subsection{Adequacy results}\label{sec:adeq}
Before proving the correctness of our implementation, we have to
define how to translate a species set back to a system of the brane
calculus.

Let $S$ be a non empty species set. A \emph{root name} of $S$, denoted
by $root(S)$, is a name $x$ such that $S(\cell{C}{^x_y}) > 0$ for some
$C,y$, and for all $z,C': \cell{C'}{^z_x}\notin dom(S)$. The next
result states that $root(\_)$ is well defined on the species sets we
encounter during a simulation.
\begin{lemma} For all $P\in\B_\sys$:
  \begin{enumerate}
  \item if $P\neq\void$: $root(\bamc{P}_x) = x$.
  \item if $\bamc{P}_x \trs{a, O} E\vdash (t,S,R)$ and
    $S\neq\emptyset$, then $root(S) = x$.
  \end{enumerate}
\end{lemma}
\begin{proof}
  (1.) is trivial by definition.  (2.) It is enough to check that the
  reaction rules introduced by $reactions(\_)$ do not change the name
  of the root, nor introduce new ones.
\end{proof}
We can now define a function $\mabc{\_}$ which maps complexes to
membranes, and a function $\mabc{\_}_x$ mapping species sets to
systems; the latter is parametric in the name $x$ of the root of the
system.
\begin{equation*}
\mabc{A_1,\dots,A_n} \triangleq A_1|\dots|A_n
\qquad
\mabc{S}_x \triangleq \prod_{\cell{C}{^x_y} \in dom(S)}
{S(\cell{C}{^x_y})} \cdot (\cell{\mabc{C}}{\mabc{S}_y})
\end{equation*}
where the notation $n\cdot P$ is a shorthand for $P\comp\dots\comp P$,
$n$ times.  

\begin{lemma} For all $P\in\B_\sys$:
  \begin{enumerate}
  \item $\mabc{\bamc{P}_x}_x \equiv P$.
  \item if $\bamc{P}_x \trs{a, O} E\vdash (t,S,R)$ then $\mabc{S}_x$ is well defined.
  \end{enumerate}
\end{lemma}
\begin{proof}
  (1.) is easy.  (2.) It is enough to check that the reaction rules
  introduced by $reactions(\_)$ do not introduce loops (i.e., the
  order among names is well founded).
\end{proof}

We can now state and prove the main results of this section.

\begin{proposition}[Soundness]\label{prop:soundness}
  For all $P\in\B_\sys$, if $\bamc{P}_x \trs{a, O} E\vdash (t,S,R)$
  then there exists $\mu$ such that $P \trs{} \mu$ and
  $\mu_{id}([\mabc{S}_x]) = a$.
\end{proposition}
\begin{proof}
  The proof is by cases on which reaction rule $O$ is selected by the
  function $next$.  By additivity of measures, we can restrict
  ourselves to when the whole process $P$ is the redex of the
  reduction.  Let us see here the case when $P$ is the redex of a
  (red-pino) (another is in Appendix~\ref{sec:longsound}).
  
  Let $P = \cell{\pino_n(\rho).\sigma|\sigma_0}{P'}$ and let us assume
  that $\sigma_0$ does not exhibit a $\pino_n$ action. Then, the
  translation of $P$ is $\bamc{P}_x = E \vdash
  species_{\emptyset,x}(P) \oplus (0,\emptyset,\emptyset)$, where
  \begin{align*}
    species_{\emptyset,x}(P) &= \{\cell{s(\pino_n(\rho).\sigma|\sigma_0)}{^x_y} \mapsto 1\} \cup species_{\{y\},x}(P')\\
    &= \{\cell{(s(\pino_n(\rho).\sigma) \cup s(\sigma_0))}{^x_y} \mapsto 1\} \cup species_{\{y\},x}(P')\\
    &= \{\cell{(\{\pino_n(\rho).\sigma\} \cup s(\sigma_0))}{^x_y} \mapsto 1\}
    \cup species_{\{y\},x}(P')
  \end{align*}
  Let $I_1 = \cell{(\{\pino_n(\rho).\sigma\} \cup
    s(\sigma_0))}{^x_y}$ and $J_{P'} = species_{\{y\},x}(P')$; then
  $\bamc{P}_x = I_1 \mapsto 1 \oplus J_{P'} \oplus (0,\emptyset,\emptyset) =
  J_{P'} \oplus (0,S', R')$ where
\begin{align*}
S' &= \{I_1 \mapsto 1\}\\
R' &= init(L, (0,S',\emptyset)) = \{O_L \mapsto (t_1,a_1)\}\\
O_L &= (\{I_1 \mapsto 1\}, r_n, id, \{\cell{(s(\sigma) \cup  s(\sigma_0))}{^x_y} \mapsto 1, \cell{s(\rho)}{^y_w} \mapsto 1\}) \\
L &= reactions(I_1 \mapsto 1,\emptyset) = unary(I_1) = \{(\{I_1 \mapsto 1\}, r_n, id, \{\cell{(s(\sigma) \cup s(\sigma_0))}{^x_y} \mapsto 1, \cell{s(\rho)}{^y_w} \mapsto 1\} )\}
\end{align*}
Now, the reaction $O$ in $\bamc{P}_x \trs{a, O} T$ is
$O_L$ (otherwise it would involve $P'$, not the pino of the whole
$P$). This means that $\bamc{P}_x \trs{a, O} T$ is derived by means of
an application of the (\ref{rl:reac}) as follows, where $S_1 = \{I_1 \mapsto 1\}$
and $S_2 = \{\cell{(s(\sigma) \cup s(\sigma_0))}{^x_y} \mapsto 1,
\cell{s(\rho)}{^y_w} \mapsto 1\}$:

		\begin{equation*}
		\rl{E \vdash \{J_{P'}\} \oplus (0, S', R') \trs{a_1, (S_1,r_n,id,S_2)} norm(E \cup fn(S_2) \vdash (\{J_{P'}\} \oplus S_2 \oplus ((t_1,S'',R'') \ominus S_1)))}%
			  {((S_1,r_n,id,S_2), a_1, t_1) = next(0,S',R') \quad (E' \vdash (0,S'',R'')) = cow(E \vdash (0,S',R'),S_1)}
		\end{equation*}

where $S'' = S'$ and $R'' = R'$.

$\{J_{P'}\} \oplus S_2 \oplus ((t_1,S'',R'') \ominus S_1) =$\\
\hspace*{30mm} $= J_{P'} \oplus \{\cell{(s(\sigma) \cup s(\sigma_0))}{^x_y} \mapsto 1,
\cell{s(\rho)}{^y_w} \mapsto 1\} \oplus ((t_1,S',R') \ominus \{I_1 \mapsto 1\})$\\
\hspace*{30mm} $= J_{P'} \oplus \{\cell{(s(\sigma) \cup s(\sigma_0))}{^x_y} \mapsto 1, \cell{s(\rho)}{^y_w} \mapsto 1\} \oplus (t_1, \{I_1 \mapsto 0\}, \{O_L \mapsto (t_2,a_2)\})$\\
\hspace*{30mm} $= J_{P'} \oplus J' \oplus (t_1, \{I_1 \mapsto 0\}, \{O_L \mapsto (t_2,a_2)\})$\\

Now let us define $Q$ as $Q =
\cell{\sigma|\sigma_0}{\cell{\rho}{\void} \comp P'}$, then $\bamc{Q}_x
= species_{\emptyset,x}(Q) \oplus (0,\emptyset,\emptyset)$ where
\begin{align*}
  species_{\emptyset,x}(Q) &= species_{\emptyset,x}(\cell{\sigma|\sigma_0}{\cell{\rho}{\void} \comp P'})\\
  &= \{\cell{s(\sigma|\sigma_0)}{^x_y} \mapsto 1\} \cup species_{\{x\},y}(\cell{\rho}{\void} \comp P')\\
  &= \{\cell{(s(\sigma) \cup s(\sigma_0))}{^x_y} \mapsto 1\} \cup species_{\{x\},y}(\cell{\rho}{\void}) \cup species_{\{x\},y}(P')\\
  &= \{\cell{(s(\sigma) \cup
    s(\sigma_0))}{^x_y} \mapsto 1\}\cup
  \{\cell{s(\rho)}{^y_w} \mapsto 1\} \cup \emptyset \cup
  J_{P'} = J' \cup J_{P'}
\end{align*}
and hence $Q \equiv\mabc{J_{P'}\oplus J' \oplus (t_1, \{I_1 \mapsto
  0\}, \{O_L \mapsto (t_2,a_2)\})}_x$.
It remains to prove that $r_n = \mu_{id}([Q])$. Let us notice that
the derivation of $P\to\mu$ is actually as follows:
\begin{prooftree}
	\AxiomC{$\pino_n(\rho).\sigma \trs{} [\pino_n]_{\sigma}^\rho$}
	\AxiomC{$\sigma_0 \trs{} \mu''$}
	\LeftLabel{(par)}
	\BinaryInfC{$\pino_n(\rho).\sigma|\sigma_0 \trs{} [\pino_n]^{\rho}_{\sigma} \mpar{\pino_n(\rho).\sigma}{\sigma_0} \mu''$}
	\AxiomC{$P' \trs{} \mu'$}
	\LeftLabel{(loc)}
	\BinaryInfC{$\cell{\pino_n(\rho).\sigma|\sigma_0}{P'} \trs{} \mu$}
\end{prooftree}
where $\mu = \mu' \mAt{\pino_n(\rho).\sigma|\sigma_0}{P'}
([\pino_n]^{\rho}_{\sigma} \mpar{\pino_n(\rho).\sigma}{\sigma_0}
\mu'')$.  Then:
\begin{align*}
\mu_{id}([\cell{\sigma|\sigma_0}{\cell{\rho}{\void} \comp P'}]) & =
 (\mu' \mAt{\pino_n(\rho).\sigma|\sigma_0}{P'} ([\pino_n]^{\rho}_{\sigma}~ \mpar{\pino_n(\rho).\sigma}{\sigma_0} \mu''))_{id}([\cell{\sigma|\sigma_0}{\cell{\rho}{\void} \comp P'}])\\
&= \mu'_{id}([\cell{\sigma|\sigma_0}{\cell{\rho}{\void} \comp P'}]) + 
 ([\pino_n]^{\rho}_{\sigma} \mpar{\pino_n(\rho).\sigma}{\sigma_0} \mu'')_{\pino_n} ([\sigma|\sigma_0] \times [\rho])\\
&= ([\pino_n]^{\rho}_{\sigma} \mpar{\pino_n(\rho).\sigma}{\sigma_0} \mu'')_{\pino_n} ([\sigma|\sigma_0] \times [\rho])\\
&= ([\pino_n]^{\rho}_{\sigma})_{\pino_n} ([\sigma|\sigma_0] \times [\rho]) + \mu''_{\pino_n} ([\sigma|\sigma_0]\times [\rho])\\
&= r_n + \mu''_{\pino_n} ([\sigma|\sigma_0] \times [\rho])  = r_n
\end{align*}
where the last equivalence holds because $\mu''_{\pino_n} ([\sigma|\sigma_0] \times
[\rho]) = 0$ because we assumed that the reaction does not involve $\sigma_0$.
\end{proof}

\begin{proposition}[Progress]\label{prop:progress}
  For all processes $P,Q$, if $P\react Q$ then there exists a reaction
  $O$ and a term $T$ such that $\bamc{P}_x \trs{a,O} T$ and $Q \equiv \mabc{T}_x$.
\end{proposition}
\begin{proof}
  By induction on the derivation of $P\react Q$.  Let us see the case
  of (red-pin), the others being similar.  Let $P =
  \cell{\pino(\rho).\sigma | \sigma_0}{P'}$ and $Q = \cell{\sigma |
    \sigma_0}{\cell{\rho}{\void} \comp P'}$.  Then, 
\[ \bamc{P}_x =  species_{\{y\},x}(P') \oplus
  (0,\{\cell{(\{\pino_n(\rho).\sigma\} \cup s(\sigma_0))}{^x_y}
  \mapsto 1\}, \{O_L \mapsto (t_1,a_1)\})
\]
where $O_L = (\{\cell{(\{\pino_n(\rho).\sigma\} \cup
  s(\sigma_0))}{^x_y} \mapsto 1\}, r_n, id, \{\cell{(s(\sigma) \cup
  s(\sigma_0))}{^x_y} \mapsto 1, \cell{s(\rho)}{^y_w} \mapsto 1 \})$.  Then, by the
(\ref{rl:reac}) we can take
$ T = species_{\{x\},y}(P') \oplus (t_1,\{ \cell{(s(\sigma) \cup
    s(\sigma_0))}{^x_y} \mapsto 1, \cell{s(\rho)}{^y_w} \mapsto 1\},
  \emptyset)
$.
It is easy to check that $Q\equiv \mabc{T}_x$.
\end{proof}

\begin{proposition}[Completeness]\label{prop:completeness}
  For all processes $P,Q$, if $P\trs{} \mu$ and $\mu_{id}([Q]) > 0$ then
  for all node name $x$, there exists a reaction $O$ and a term $T$
  such that $\bamc{P}_x \trs{a,O} T$, $Q\equiv \mabc{T}_x$ and
  $a=\mu_{id}([Q])$.
\end{proposition}
\begin{proof}
  If $P\trs{} \mu$ and $\mu_{id}([Q]) > 0$ then $P\react Q$ by
  Prop.~\ref{prop:conservative}.  By Prop.~\ref{prop:progress}, we
  have that for some $a,O,T$, $\bamc{P}_x \trs{a,O} T$ and $Q\equiv
  \mabc{T}_x$. But then $a=\mu_{id}([Q])$ by soundness
  (Prop.~\ref{prop:soundness}).
\end{proof}

\subsection{Example}
We conclude this section with an example.
Let $P = 10000 \cdot \cell{\phago_n.\exo_m}{\cell{\phago_k}{}} \comp
100 \cdot (\cell{(\cophago_n(\coexo_m) \mid
\coexo)}{\cell{\phago_k}{}})$. Then, its reductions in the
Brane Calculus are as follows:
\begin{align*}
&10000 \cdot \cell{\phago_n.\exo_m}{\cell{\phago_k}{}} \comp 100 \cdot (\cell{(\cophago_n(\coexo_m) \mid \coexo)}{\cell{\phago_k}{}})\\
&\react 9999 \cdot \cell{\phago_n.\exo_m}{\cell{\phago_k}{}} \comp 99 \cdot (\cell{(\cophago_n(\coexo_m) \mid \coexo)}{\cell{\phago_k}{}}) \comp \cell{\coexo}{\cell{\coexo_m}{\cell{\exo_m}{\cell{\phago_k}{}}} \comp \cell{\phago_k}{}} \\
&\react 9999 \cdot \cell{\phago_n.\exo_m}{\cell{\phago_k}{}} \comp 99 \cdot (\cell{(\cophago_n(\coexo_m) \mid \coexo)}{\cell{\cophago_k()}{}}) \comp \cell{\coexo}{2 \cdot \cell{\cell{\phago_k}{} \comp \void}{\void} }
\end{align*}
The translation of $P$ is $\bamc{P}_x = E \vdash
species_{\emptyset,x}(P) \oplus (0,\emptyset,\emptyset)$, where
\begin{align*}
species_{\emptyset,x}(P) &= species_{\{x\}}(10000 \cdot
\cell{\phago_n.\exo_m}{\cell{\phago_k}{}}
\circ 100 \cdot \cell{(\cophago_n(\coexo_m) \mid \coexo)}{\cell{\phago_k}{}})\\
&= \{s(\cell{\phago_n.\exo_m)}{^x_y} \mapsto 10000\} \cup
species_{\{x,y\},y}(\cell{\phago_k}{}) \cup
\{s\cell{(\cophago_n(\coexo_m) \mid \coexo)}{^x_z} \mapsto 100\}\\
&\phantom{= }~ \cup species_{\{x,z\},z}(\cell{\phago_k}{})\\
&= \{\cell{\{\phago_n.\exo_m\}}{^x_y} \mapsto 10000, \cell{\{\cophago_n(\coexo_m), \coexo\}}{^x_z} \mapsto 100, \cell{\{\phago_k\}}{^y_w} \mapsto 1, \\
&\phantom{= }~ \cell{\{\phago_k\}}{^z_v} \mapsto 1\}
\end{align*}
Let $I_1 = \cell{\{\phago_n.\exo_m\}}{^x_y}$, $I_2 =
\cell{\{\cophago_n(\coexo_m), \coexo\}}{^x_z}$, $I_3 =
\cell{\{\phago_k\}}{^y_w}$, $I_4
=\cell{\{\phago_k\}}{^z_v}$, $r_n = 10 s^{-1}$, $r_k = 5 s^{-1}$ and $r_m = 5 s^{-1}$; then
\begin{align*}
\bamc{P}_x & = I_1 \mapsto 10000 \oplus I_2 \mapsto 100 \oplus
I_3 \mapsto 1 \oplus I_4 \mapsto 1 \oplus (0,\emptyset,\emptyset) \\
& = I_2 \mapsto 100 \oplus I_3 \mapsto 1 \oplus I_4 \mapsto 1 \oplus (0,S_1,R_1) \\
& = I_3 \mapsto 1 \oplus I_4 \mapsto 1 \oplus (0,S_2,R_2)\\
& = I_4 \mapsto 1 \oplus (0,S_3,R_3) = (0,S_4,R_4)
\end{align*}
where
\begin{align*}
L_1 &= reactions(I_1 \mapsto 10000, \emptyset) = \emptyset \\
S_1 &= S\{I_1 \mapsto 10000\} \\
R_1 &= R \cup init(L_1, (0,S_1,\emptyset)) = \emptyset
\\[1ex]
L_2 &= reactions(I_2 \mapsto 100,S_1) \\
&= (\{I_2 \mapsto 1,I_1 \mapsto 1\}, 10, id, \{\cell{\{\coexo\}}{^x_z} \mapsto 1,
\cell{\{\coexo_m)\}}{^z_w} \mapsto 1, \cell{\{\exo_m\}}{^w_y} \mapsto 1\})
\\
S_2 &= S_1\{I_2 \mapsto 100\} = \{I_1 \mapsto 10000, I_2 \mapsto 100\}
\\
R_2 &= R_1 \cup init(L_2, (0,S_2,R_1)) = \{O_{L_2} \mapsto (t_1,a_1)\}
\\[1ex]
L_3 &= reactions(I_3 \mapsto 1, S_2) = \emptyset
\\
S_3 &= S_2\{I_3 \mapsto 1\} = \{I_1 \mapsto 10000, I_2 \mapsto 100, I_3 \mapsto 1\}
\\
R_3 &= R_2 \cup init(L_3, (0,S_3,R_2)) = R_2
\\[1ex]
L_4 &= reactions(I_4 \mapsto 1, S_3) = \emptyset
\\
S_4 &= S_3\{I_4 \mapsto 1\} = \{I_1 \mapsto 10000, I_2 \mapsto 100, I_3 \mapsto 1, I_4 \mapsto 1\}
\\
R_4 &= R_3 \cup init(L_4, (0,S_4,R_3)) = R_3
\end{align*}
with
$\begin{array}[t]{rl}
O_{L_2} & = (\{I_2 \mapsto 1,I_1 \mapsto 1\}, 10, id, \{\cell{\{\coexo\}}{^x_z} \mapsto 1,
\cell{\{\coexo_m\}}{^z_w} \mapsto 1, \cell{\{\exo_m\}}{^w_y} \mapsto 1\}) \\
a_1 &= propensity(O_{L_2},S_2) = 10000000 \\
t_1 &= 0 + delay(10,10000000).
\end{array}$
\\
Now, the reaction $O$ in $\bamc{P}_x \trs{a, O} T$ is $O_L$. This
means that $\bamc{P}_x \trs{a, O} T$ is derived by means of an
application of the (\ref{rl:reac}) as follows, where $E = x,y,z,w$,
$S_1 = \{I_2 \mapsto 1,I_1 \mapsto 1\}$, and $S_2 =
\{\cell{\{\coexo\}}{^x_z} \mapsto 1,
\cell{\{\coexo_m\}}{^z_w} \mapsto 1, \cell{\{\exo_m\}}{^w_y} \mapsto 1\}$:
\begin{equation*}
\rl{\begin{array}{l}
E \vdash
(0, S_4, R_4) \trs{10, (S_1,10,id,S_2)} norm(E' \cup fn(S_2) \vdash
(S_2
\oplus ((t_1,S_5,R_5) \ominus S_1)))
\end{array}
}%
{((S_1,10,id,S_2), 10, t_1) = next(0,S_4,R_4) \quad (E' \vdash (0,S_5,R_5)) = cow(E \vdash (0,S_4,R_4),S_1)}
\end{equation*}
where
\begin{align*}
S_5 &= \{I_1 \mapsto 1, \cell{\{\phago_n.\exo_m\}}{^x_{y'}} \mapsto 9999, I_2 \mapsto 1, \cell{\{\cophago_n(\coexo_m),\coexo\}}{^x_{z'}} \mapsto 99, I_3 \mapsto 1, \\
& \phantom{\ = \{} I_4 \mapsto 1, \cell{\{\phago_k\}}{^{y'}_{w'}} \mapsto 1, \cell{\{\phago_k\}}{^{z'}_{v'}} \mapsto 1\}
\\
R_5 &= \{O_{L_2} \mapsto (t_1,a_1), O_1  \mapsto (t_2,a_2), O_2 \mapsto (t_3,a_3), O_3 \mapsto (t_4,a_4)\}
\intertext{with}
O_1 &= (\{\cell{\{\phago_n.\exo_m\}}{^x_{y'}} \mapsto 1,I_2 \mapsto 1\}, 10, id, 
\{\cell{\{\coexo\}}{^x_z} \mapsto 1,
\cell{\{\coexo_m\}}{^z_{w''}} \mapsto 1, \cell{\{\exo_m\}}{^{w''}_{y'}} \mapsto 1\})\\
O_2 &= (\cell{\{\cophago_n(\coexo_m),\coexo\}}{^x_{z'}} \mapsto 1,I_1 \mapsto 1\}, 10, id, \{\cell{\{\coexo\}}{^x_{z'}} \mapsto 1,
\cell{\{\coexo_m\}}{^{z'}_{w'''}} \mapsto 1, \cell{\{\exo_m\}}{^{w'''}_{y}} \mapsto 1\})\\
O_3 &= (\{\cell{\{\phago_n.\exo_m\}}{^x_{y'}} \mapsto 1, \cell{\{\cophago_n(\coexo_m),\coexo\}}{^x_{z'}} \mapsto 1\},10,id,
\{\cell{\{\cophago\}}{^x_{z'}} \mapsto 1, \cell{\{\coexo_m\}}{^{z'}_{z''}} \mapsto 1, \cell{\{\exo_m\}}{^{z''}_{y'}}\})\\
a_2 &= propensity(O_1,S_5) = 99990 \\
a_3 &= propensity(O_2,S_5) = 990 \\
a_4 &= propensity(O_3,S_5) = 989901 \\
t_2 &= 0 + delay(10,99990)\\
t_3 &= 0 + delay(10,990)\\
t_4 &= 0 + delay(10,989901)
\end{align*}
We can now compute the multiset of the new machine state:
\begin{align*}
S_2 \oplus ((t_1,S_5,R_5) \ominus S_1) &= S_2 \oplus ((t_1,S_5,R_5) \ominus \{I_2 \mapsto 1,I_1 \mapsto 1\})\\
&= S_2 \oplus ((t_1,S_6, R_6\}) \ominus \{I_1 \mapsto 1\})\\
&= \{\cell{\{\coexo\}}{^x_z} \mapsto 1,
\cell{\{\coexo_m\}}{^z_w} \mapsto 1, \cell{\{\exo_m\}}{^w_y} \mapsto 1\}
\oplus (t_1,S_7, R_7)\\
&= (t_1, S_8, R_8)
\end{align*}
with
\begin{alignat*}{3}
a'_2 &= propensity(O_1,S_6) = 0
&
S_8 &= S_6 \cup \{\cell{\{\coexo\}}{^x_z} \mapsto 1,
 \cell{\{\coexo_m\}}{^z_w} \mapsto 1, \cell{\{\exo_m\}}{^w_y} \mapsto 1\}
\\
t'_2 &= t_1 + (a_2 / a'_2)(t_2 - t_1)
&
R_6 &= \{O_1 \mapsto (t'_2,a'_2)\}
\\%
a'_3 &= propensity(O_2,S_7) = 0 \qquad\quad
&
R_7 &= R_6 \cup \{O_2 \mapsto (t'_3,a'_3)\}
\\
t'_3 &= t_1 + (a_3 / a_5)(t_3 - t_1)
&
R_8 &= R_7 \cup \{\{\cell{\{\coexo_m\}}{^z_w} \mapsto 1, \cell{\{\exo_m\}}{^w_y} \mapsto 1\},5,f,\{\cell{\emptyset}{^z_w}\} \mapsto 1\}
\\
S_6 &= S_5 \backslash \{I_2 \mapsto 1\}\\
S_7 &= S_6 \backslash \{I_1 \mapsto 1\}
&
f &= \lambda T.T[y := z]
\end{alignat*}

\section{Conclusions}\label{sec:concl}
In this paper, we have presented an abstract machine for the
Stochastic Brane Calculus.  Instead of defining an \emph{ad hoc}
machine, we have adopted the generic abstract machine for stochastic
calculi (GSAM) recently introduced by Lakin, Paulev\'e and Phillips.
According to the encoding technique we have adopted, membranes are
flattened into a set of species, where the hierarchical structure is
represented by means of \emph{names}. In order to keep track of these
names, and for dealing efficiently with multiple copies of the same
species, we have introduced a new generic abstract machine, called
COWGSAM, which extends the GSAM with a \emph{name environment} and a
\emph{copy-on-write} optimization strategy.  We have proved that the
implementation of the Stochastic Brane Calculus in COWGSAM is adequate
with respect to the stochastic structural operational semantics of the
calculus given in \cite{bm:tcs12}.

We think that COWGSAM can be used for implementing other stochastic
calculi dealing with nested structures, also beyond the models for
systems biology.  In particular, it is interesting to apply this
approach to Stochastic Bigraphs \cite{kmt:mfps08}, a general
meta-model well-suited for representing a range of stochastic systems
with compartments; in this way we would obtain a \emph{General
  Stochastic Bigraphical Machine}, which could be instantiated to any
given stochastic bigraphic reactive system. However, such a machine
would not scale well, as in general the COW strategy may be not very
useful; thus, we can restrict our attentions to smaller subsets of
BRSs, specifically designed to some application domain.  For
biological applications, the bigraphic reactive systems considered in
\cite{bgm:biobig,dhk:fcm} might be a more reasonable target.

Another interesting question is about the expressive power of GSAM and
COWGSAM. We think that GSAM correspond to stochastic (multiset) Petri
nets, but COWGSAM could go further thanks to the possibility of
creating unlimited new names during execution.
Further work include comparison with other stochastic simulation tools
dealing with compartments, like BioPEPA \cite{cg:biopepa}.

\smallskip

\noindent\textbf{Acknowledgment}
Work funded by MIUR PRIN project ``SisteR'',   prot.~20088HXMYN.

\vspace{-1ex}

\bibliographystyle{eptcs}
\bibliography{allbib}

\begin{thebibliography}{10}
\providecommand{\bibitemstart}[1]{\bibitem{#1}}
\providecommand{\bibitemend}{}
\providecommand{\bibliographystart}{}
\providecommand{\bibliographyend}{}
\providecommand{\url}[1]{\texttt{#1}}
\providecommand{\urlprefix}{doi:}
\providecommand{\bibinfo}[2]{#2}
\providecommand{\doiurl}[1]{\urlprefix\href{http://dx.doi.org/#1}{#1}}
\bibliographystart

\bibitemstart{bgm:biobig}
\bibinfo{author}{Giorgio Bacci}, \bibinfo{author}{Davide Grohmann} \&
  \bibinfo{author}{Marino Miculan} (\bibinfo{year}{2009}):
  \emph{\bibinfo{title}{Bigraphical models for protein and membrane
  interactions}}.
\newblock In: \bibinfo{editor}{Gabriel Ciobanu}, editor: {\sl
  \bibinfo{booktitle}{Proc.~MeCBIC'09}}, {\sl
  \bibinfo{series}{EPTCS}}~\bibinfo{volume}{11}, pp. \bibinfo{pages}{3--18}.
\newblock \doiurl{10.4204/EPTCS.11.1}.
\bibitemend

\bibitemstart{bm:tcs12}
\bibinfo{author}{Giorgio Bacci} \& \bibinfo{author}{Marino Miculan}
  (\bibinfo{year}{2012}): \emph{\bibinfo{title}{Measurable Stochastics for
  {Brane Calculus}}}.
\newblock {\sl \bibinfo{journal}{Theoretical Computer Science}}
  \bibinfo{volume}{431}, pp. \bibinfo{pages}{117--136}.
\newblock \doiurl{10.1016/j.tcs.2011.12.055}.
\bibitemend

\bibitemstart{bmsmt:cls}
\bibinfo{author}{Roberto Barbuti}, \bibinfo{author}{Andrea Maggiolo-Schettini},
  \bibinfo{author}{Paolo Milazzo} \& \bibinfo{author}{Angelo Troina}
  (\bibinfo{year}{2007}): \emph{\bibinfo{title}{The Calculus of Looping
  Sequences for Modeling Biological Membranes}}.
\newblock In: \bibinfo{editor}{George Eleftherakis}, \bibinfo{editor}{Petros
  Kefalas}, \bibinfo{editor}{Gheorghe Paun}, \bibinfo{editor}{Grzegorz
  Rozenberg} \& \bibinfo{editor}{Arto Salomaa}, editors: {\sl
  \bibinfo{booktitle}{Workshop on Membrane Computing}}, {\sl
  \bibinfo{series}{Lecture Notes in Computer Science}} \bibinfo{volume}{4860},
  \bibinfo{publisher}{Springer}, pp. \bibinfo{pages}{54--76}.
\newblock \doiurl{10.1007/978-3-540-77312-2\_4}.
\bibitemend

\bibitemstart{cardelli04:bc}
\bibinfo{author}{Luca Cardelli} (\bibinfo{year}{2004}):
  \emph{\bibinfo{title}{Brane Calculi}}.
\newblock In: \bibinfo{editor}{Vincent Danos} \& \bibinfo{editor}{Vincent
  Sch{\"a}chter}, editors: {\sl \bibinfo{booktitle}{Proc.~CMSB}}, {\sl
  \bibinfo{series}{Lecture Notes in Computer Science}} \bibinfo{volume}{3082},
  \bibinfo{publisher}{Springer}, pp. \bibinfo{pages}{257--278}.
\newblock \doiurl{10.1007/978-3-540-25974-9\_24}.
\bibitemend

\bibitemstart{cm:quest10}
\bibinfo{author}{Luca Cardelli} \& \bibinfo{author}{Radu Mardare}
  (\bibinfo{year}{2010}): \emph{\bibinfo{title}{The Measurable Space of
  Stochastic Processes}}.
\newblock In: {\sl \bibinfo{booktitle}{Proc.~QEST}}, \bibinfo{publisher}{IEEE
  Computer Society}, pp. \bibinfo{pages}{171--180}.
\newblock \doiurl{10.1109/QEST.2010.30}.
\bibitemend

\bibitemstart{cg:biopepa}
\bibinfo{author}{Federica Ciocchetta} \& \bibinfo{author}{Maria~Luisa
  Guerriero} (\bibinfo{year}{2009}): \emph{\bibinfo{title}{Modelling Biological
  Compartments in {Bio-PEPA}}}.
\newblock {\sl \bibinfo{journal}{Electronic Notes in Theoretical Computer
  Science}} \bibinfo{volume}{227}, pp. \bibinfo{pages}{77--95}.
\newblock \doiurl{10.1016/j.entcs.2008.12.105}.
\bibitemend

\bibitemstart{dhk:fcm}
\bibinfo{author}{Troels~Christoffer Damgaard}, \bibinfo{author}{Espen
  H{\o}jsgaard} \& \bibinfo{author}{Jean Krivine} (\bibinfo{year}{2012}):
  \emph{\bibinfo{title}{Formal Cellular Machinery}}.
\newblock {\sl \bibinfo{journal}{Electronic Notes in Theoretical Computer
  Science}} \bibinfo{volume}{284}, pp. \bibinfo{pages}{55--74}.
\newblock \doiurl{10.1016/j.entcs.2012.05.015}.
\bibitemend

\bibitemstart{hermanns02}
\bibinfo{author}{Holger Hermanns} (\bibinfo{year}{2002}):
  \emph{\bibinfo{title}{Interactive Markov Chains: The Quest for Quantified
  Quality}}, {\sl \bibinfo{series}{Lecture Notes in Computer Science}}
  \bibinfo{volume}{2428}.
\newblock \bibinfo{publisher}{Springer}.
\newblock \doiurl{10.1007/3-540-45804-2}.
\bibitemend

\bibitemstart{kmt:mfps08}
\bibinfo{author}{Jean Krivine}, \bibinfo{author}{Robin Milner} \&
  \bibinfo{author}{Angelo Troina} (\bibinfo{year}{2008}):
  \emph{\bibinfo{title}{Stochastic Bigraphs}}.
\newblock In: {\sl \bibinfo{booktitle}{Proc.~24th MFPS}}, {\sl
  \bibinfo{series}{ENTCS}} \bibinfo{volume}{218}, pp. \bibinfo{pages}{73--96}.
\newblock \doiurl{10.1016/j.entcs.2008.10.006}.
\bibitemend

\bibitemstart{lpp:tcs12}
\bibinfo{author}{Matthew~R. Lakin}, \bibinfo{author}{Lo{\"i}c Paulev{\'e}} \&
  \bibinfo{author}{Andrew Phillips} (\bibinfo{year}{2012}):
  \emph{\bibinfo{title}{Stochastic Simulation of Multiple Process Calculi for
  Biology}}.
\newblock {\sl \bibinfo{journal}{Theoretical Computer Science}}
  \bibinfo{volume}{431}, pp. \bibinfo{pages}{181--206}.
\newblock \doiurl{10.1016/j.tcs.2011.12.057}.
\bibitemend

\bibitemstart{lt:biokappa}
\bibinfo{author}{Cosimo Laneve} \& \bibinfo{author}{Fabien Tarissan}
  (\bibinfo{year}{2008}): \emph{\bibinfo{title}{A simple calculus for proteins
  and cells}}.
\newblock {\sl \bibinfo{journal}{Theoretical Computer Science}}
  \bibinfo{volume}{404}(\bibinfo{number}{1-2}), pp. \bibinfo{pages}{127--141}.
\newblock \doiurl{10.1016/j.tcs.2008.04.011}.
\bibitemend

\bibitemstart{prakashbook09}
\bibinfo{author}{Prakash Panangaden} (\bibinfo{year}{2009}):
  \emph{\bibinfo{title}{Labelled Markov Processes}}.
\newblock \bibinfo{publisher}{Imperial College Press},
  \bibinfo{address}{London, U.K.}
\bibitemend

\bibitemstart{pylp:cmsb10}
\bibinfo{author}{Lo\"{\i}c Paulev{\'e}}, \bibinfo{author}{Simon Youssef},
  \bibinfo{author}{Matthew~R. Lakin} \& \bibinfo{author}{Andrew Phillips}
  (\bibinfo{year}{2010}): \emph{\bibinfo{title}{A generic abstract machine for
  stochastic process calculi}}.
\newblock In: \bibinfo{editor}{Paola Quaglia}, editor: {\sl
  \bibinfo{booktitle}{Proc.~CMSB}}, \bibinfo{publisher}{ACM}, pp.
  \bibinfo{pages}{43--54}.
\newblock \doiurl{10.1145/1839764.1839771}.
\bibitemend

\bibitemstart{rpscs:bioambients}
\bibinfo{author}{Aviv Regev}, \bibinfo{author}{Ekaterina~M. Panina},
  \bibinfo{author}{William Silverman}, \bibinfo{author}{Luca Cardelli} \&
  \bibinfo{author}{Ehud~Y. Shapiro} (\bibinfo{year}{2004}):
  \emph{\bibinfo{title}{BioAmbients: an abstraction for biological
  compartments}}.
\newblock {\sl \bibinfo{journal}{Theoretical Computer Science}}
  \bibinfo{volume}{325}(\bibinfo{number}{1}), pp. \bibinfo{pages}{141--167}.
\newblock \doiurl{10.1016/j.tcs.2004.03.061}.
\bibitemend

\bibliographyend
\end{thebibliography}

\appendix

\section{Some measure theory}
\label{app:measuretheory}

Given a set $M$, a family $\Sigma$ of subsets of $M$ is called
a \emph{$\sigma$-algebra} if it contains $M$ and is closed under 
complements and (infinite) countable unions:
\begin{enumerate}[itemsep=0pt]
  \item $M \in \Sigma$;
  \item $A \in \Sigma$ implies $A^{c} \in \Sigma$, 
    where $A^{c} = M \setminus A$;
  \item $\set{A_i}_{i \in \N} \subset \Sigma$ implies 
    $\bigcup_{i \in \N} A_i \in \Sigma$.
\end{enumerate}
Since $M \in \Sigma$ and $M^{c} = \emptyset$, $\emptyset \in \Sigma$,
hence $\Sigma$ is nonempty by definition.
A $\sigma$-algebra is closed under countable set-theoretic operations:
is closed under finite unions ($A, B \in \Sigma$ implies $A \cup B = A \cup B \cup \emptyset \cup \emptyset \cup \dots \in \Sigma$), countable intersections (by DeMorgan's law $A \cap B = (A^{c} \cup B^{c})^{c}$ 
in its finite and inifite version), and countable subtractions 
($A, B \in \Sigma$ implies $A \setminus B = A \cap B^{c} \in \Sigma$).

\begin{definition}[Measurable Space]
  Given a set $M$ and a $\sigma$-algebra on $M$, the tuple $(M,
  \Sigma)$ is called a \emph{measurable space}, the elements of
  $\Sigma$ \emph{measurable sets}, and $M$ the \emph{support-set}.
\end{definition}

\looseness=-1 A set $\Omega \subseteq 2^{M}$ is a \emph{generator for
  the $\sigma$-algebra $\Sigma$} on $M$ if $\Sigma$ is the closure of
$\Omega$ under complement and countable union; we write
$\sigma(\Omega) = \Sigma$ and say that $\Sigma$ is generated by
$\Omega$.  Note that the $\sigma$-algebra generated by a $\Omega$ is
also the smallest $\sigma$-algebra containing $\Omega$, that is, the
intersection of all $\sigma$-algebras that contain $\Omega$.  In
particular it holds that a completely arbitrary intersection of
$\sigma$-algebras is a $\sigma$-algebra.  A $\sigma$-algebra generated
by $\Omega$, denoted by $\sigma(\Omega)$, is minimal in the sense that
if $\Omega \subset \Sigma$ and $\Sigma$ is a $\sigma$-algebra, then
$\sigma(\Omega) \subset \Sigma$.  If $\Omega$ is a $\sigma$-algebra
then obviously $\sigma(\Omega) = \Omega$; if $\Omega$ is empty or
$\Omega = \set{\emptyset}$, or $\Omega = \set{M}$, then
$\sigma(\Omega) = \set{\emptyset, M}$; if $\Omega \subset \Sigma$ and
$\Sigma$ is a $\sigma$-algebra, then $\sigma(\Omega) \subset \Sigma$.
A generator $\Omega$ for $\Sigma$ is a \emph{base for $\Sigma$} if it
has disjoin elements. Note that if $\Omega$ is a base for $\Sigma$,
all measurable sets in $\Sigma$ can be decomposed into countable
unions of elements in $\Omega$.

A \emph{measure} on a measurable space $(M, \Sigma)$ is a function
$\mu \colon \Sigma \to \R^{+}_{\infty}$, where $\R^{+}_{\infty}$ denotes
the extended positive real line, such that
\begin{enumerate}[itemsep=0pt]
  \item $\mu(\emptyset) = 0$;
  \item for any disjoint sequence $\set{N_i}_{i \in I} \subseteq \Sigma$ with
  $I \subseteq \N$, it holds
  \begin{equation*}
    \textstyle \mu(\bigcup_{i \in I} N_i) = \sum_{i \in I} \mu(N_i) \, .
  \end{equation*}
\end{enumerate}

The triple $(M, \Sigma, \mu)$ is called a \emph{measure space}.
A measure space $(M, \Sigma, \mu)$ is called \emph{finite} if 
$\mu(M)$ is a finite real number; it is called \emph{$\sigma$-finite} if $M$ can 
be decomposed into a countable union of measurable sets of finite measure,
that is, $M = \bigcup_{i \in I} N_{i}$, for some $I \subseteq \N$ and 
$\mu(N_{i}) \in \R^{+}$ for each $i \in I$.
A set in a measure space has \emph{$\sigma$-finite measure} 
if it is a countable union of sets with finite measure.
Specifying a measure includes specifying its domain. If $\mu$ is a 
measure on a measurable space $(M, \Sigma)$ and $\Sigma'$ is
a $\sigma$-algebra contained in $\Sigma$, then the restriction $\mu'$ of 
$\mu$ to $\Sigma'$ is also a measure, and in particular a measure
on $(M', \Sigma')$, for some $M' \subseteq M$ such that $\Sigma'$ is
a $\sigma$-algebra on $M'$.

Given two measurable spaces and measures on them, one can obtain 
the product measurable space and the product measure on that space. 
Let $(M_{1}, \Sigma_{1})$ and $(M_{2}, \Sigma_{2})$ be measurable spaces,
and $\mu_{1}$ and $\mu_{2}$ be measures on these spaces.
Denote by $\Sigma_{1} \otimes \Sigma_{2}$ the $\sigma$-algebra on the 
cartesian product $M_{1} \times M_{2}$ generated by subsets of the form 
$B_{1} \times B_{2}$, said rectangles, where $B_{1} \in \Sigma_{1}$ and $B_{2} \in \Sigma_{2}$.
The \emph{product measure} $\mu_{1} \otimes \mu_{2}$ is defined to be the unique
measure on the measurable space 
$(M_{1} \times M_{2}, \Sigma_{1} \otimes \Sigma_{2})$ such that, for all 
$B_{1} \in \Sigma_{1}$ and $B_{2} \in \Sigma_{2}$
\begin{equation*}
  (\mu_{1} \otimes \mu_{2})(B_{1} \times B_{2}) = \mu_{1}(B_{1}) \cdot \mu_{2}(B_{2})
\end{equation*}
The existence of this measure is guaranteed by the Hahn-Kolmogorov theorem.
The uniqueness of the product measure is guaranteed only in the case that both
$(M_{1}, \Sigma_{1}, \mu_{1})$ and $(M_{2}, \Sigma_{2}, \mu_{2})$ are
$\sigma$-finite.

Let $\Delta(M,\Sigma)$ be the family of measures on $(M, \Sigma)$.
It can be organized as a measurable space by considering the
$\sigma$-algebra generated by the sets 
$\set{\mu \in \Delta(M, \Sigma) : \mu(S) \geq r}$, for arbitrary
$S \in \Sigma$ and $r > 0$.

Given two measurable spaces $(M, \Sigma)$ and $(N, \Theta)$ a mapping
$f \colon M \to N$ is \emph{measurable} if for any $T \in \Theta$, 
$f^{-1}(T) \in \Sigma$. Measurable functions are closed under composition:
given $f\colon M \to N$ and $g\colon N \to O$ measurable functions then
$g \circ f \colon M \to O$ is also measurable.

\vspace{-1ex}

\section{Proof of Prop.~\ref{prop:soundness}}\label{sec:longsound}

Let $P = \cell{\coexo_n.\tau|\tau_0}{\cell{\exo_n.\sigma|\sigma_0}{P'} \comp P''}$; then, $\bamc{P}_x = E \vdash species_{\emptyset,x}(P) \oplus (0,\emptyset,\emptyset)$, where
\begin{align*}
  & species_{\emptyset,x}(P) = \{\cell{s(\coexo_n.\tau|\tau_0)}{^x_y}\} \cup species_{\{x\},y}(\cell{\exo_n.\sigma|\sigma_0}{P'} \comp P'')\\
				&= \{\cell{(\{\coexo_n.\tau\} \cup s(\tau_0))}{^x_y} \mapsto 1\} \cup species_{\{x\},y}(\cell{\exo_n.\sigma|\sigma_0}{P'}) \cup species_{\{x\},y}(P'')\\
				&= \{\cell{(\{\coexo_n.\tau\} \cup s(\tau_0))}{^x_y} \mapsto 1\} \cup \{\cell{s(\exo_n.\sigma|\sigma_0)}{^y_w} \mapsto 1\} \cup species_{\{x,y\},w}(P') \cup species_{\{x\},y}(P'') \\
				&= \{\cell{(\{\coexo_n.\tau\} \cup s(\tau_0))}{^x_y} \mapsto 1\} \cup \{\cell{(s(\exo_n.\sigma) \cup s(\sigma_0))}{^y_w} \mapsto 1\} \cup species_{\{x,y\},w}(P') \cup species_{\{x\},y}(P'') \\
				&= \{\cell{(\{\coexo_n.\tau\} \cup s(\tau_0))}{^x_y} \mapsto 1\} \cup \{\cell{(\{\exo_n.\sigma\} \cup s(\sigma_0))}{^y_w} \mapsto 1\} \cup species_{\{x,y\},w}(P') \cup species_{\{x\},y}(P'')
				\end{align*}
Let $I_1 = \cell{(\{\coexo_n.\tau\} \cup s(\tau_0))}{^x_y}$, $I_2 =
\cell{(\{\exo_n.\sigma\} \cup s(\sigma_0))}{^y_w}$, $J_{P'} =
species_{\{x,y\},w}(P')$, $J_{P''} = species_{\{x\},y}(P'')$; then
$\bamc{P}_x = I_1 \mapsto 1 \oplus I_2 \mapsto 1 \oplus J_{P'} \oplus
J_{P''} \oplus (0,\emptyset,\emptyset) = I_2 \mapsto 1 \oplus J_{P'}
\oplus J_{P''} \oplus (0,S',R') = J_{P'} \oplus J_{P''} \oplus
(0,S'',R'')$ where
\vspace{-1ex}
\begin{align*}
  L' &= reactions(I_1 \mapsto 1, \emptyset) = \emptyset\\
  S' &= S\{I_1 \mapsto 1\}\\
  R' &= init(L', (0,S',\emptyset)) = \emptyset
  \\[1ex]
  L'' &= reactions(I_2 \mapsto 1, S')\\
  &= (\{I_2 \mapsto 1,I_1 \mapsto 1\}, r_n, f, \{\cell{(s(\tau) \cup s(\tau_0) \cup s(\sigma) \cup s(\sigma_0))}{^x_y} \mapsto 1\})\\
  S'' &= S'\{I_2 \mapsto 1\} = \{I_1 \mapsto 1, I_2 \mapsto 1\}\\
  R'' &= init(L'', (0,S'',R')) = \{O_L \mapsto (t_1,a_1)\}
\end{align*}
with $O_L = (\{I_1 \mapsto 1,I_2 \mapsto 1\}, r_n, f, \{\cell{(s(\tau) \cup s(\tau_0) \cup s(\sigma) \cup s(\sigma_0))}{^x_y} \mapsto 1\})$.
Now, the reaction $O$ in $\bamc{P}_x \trs{F, O} T$ is $O_L$. This means that $\bamc{P}_x \trs{F, O} T$ is derived by means of an application of the (\ref{rl:reac}) as follows, 
where $S_1 = \{I_1 \mapsto 1,I_2 \mapsto 1\}$, $S_2 = \{\cell{(s(\tau) \cup s(\tau_0) \cup s(\sigma) \cup s(\sigma_0))}{^x_y} \mapsto 1\}$ and $f = T[w := x]$:
				
\begin{equation*}
\rl{\begin{array}{l}
E \vdash
(0, S'', R'') \trs{a_1, (S_1,r_n,f,S_2)} norm(E' \cup fn(S_2) \vdash
(f(S_2 \oplus ((t_1,S''',R''') \ominus S_1))))
\end{array}
}%
{((S_1,r_n,f,S_2), a_1, t_1) = next(0,S'',R'') \quad (E' \vdash (0,S''',R''')) = cow(E \vdash (0,S'',R''),S_1)}
\end{equation*}
where $S''' = S''$ and $R''' = R''$.
\begin{align*}
& \hspace{-1cm} (f(S_2 \oplus ((t_1,S''',R''') \ominus S_1)))) =\\
			&= f((\cell{s(\tau) \cup s(\tau_0) \cup s(\sigma) \cup s(\sigma_0)}{^x_y}) \oplus ((t_1,S''',R''') \ominus \{I_1 \mapsto 1,I_2 \mapsto 1\})\\
			&= f((\cell{s(\tau) \cup s(\tau_0) \cup s(\sigma) \cup s(\sigma_0)}{^x_y}) \oplus ((t_1, \{I_1 \mapsto 0, I_2 \mapsto 1\}, \{O_L \mapsto (t_2,a_2)\}) \ominus \{I_2 \mapsto 1\})\\
			&= f((\cell{s(\tau) \cup s(\tau_0) \cup s(\sigma) \cup s(\sigma_0)}{^x_y}) \oplus (t_1, \{I_1 \mapsto 0, I_2 \mapsto 0\}, \{O_L \mapsto (t_3,a_3)\})\\
			&= (\cell{s(\tau) \cup s(\tau_0) \cup s(\sigma) \cup s(\sigma_0)}{^x_y}) \oplus (t_1, \{I_1 \mapsto 0, I_2 \mapsto 0\}, \{O_L \mapsto (t_3,a_3)\})\\
			&=  J' \oplus (t_1, \{I_1 \mapsto 0, I_2 \mapsto 0\}, \{O_L \mapsto (t_3,a_3)\})
			\end{align*}
Now let us define $Q$ as $Q = \cell{\sigma|\sigma_0|\tau|\tau_0}{P''} \comp P'$, then $\bamc{Q}_x = species_{\emptyset,x}(Q) \oplus (0,\emptyset,\emptyset)$ where
			\begin{align*}
			species_{\emptyset,x}(Q) &= species_{\emptyset,x}(\cell{\sigma|\sigma_0|\tau|\tau_0}{P''} \comp P')\\
			&= species_{\emptyset,x}(\cell{\sigma|\sigma_0|\tau|\tau_0}{P''}) \cup species_{\emptyset,x}(P')\\
			&= \{s(\cell{\sigma|\sigma_0|\tau|\tau_0)}{^x_y}) \mapsto 1\} \cup species_{\{x\},y}(P'') \cup species_{\emptyset,x}(P')\\
			&= \{\cell{s(\tau) \cup s(\tau_0) \cup s(\sigma) \cup s(\sigma_0)}{^x_y} \mapsto 1\} \cup J_{P''} \cup J_{P'}
			 = J' \cup J_{P''} \cup J_{P'}
			\end{align*}
And hence $Q \equiv \mabc{J_{P'} \oplus J_{P''} \oplus J' \oplus (t_1, \{I_1 \mapsto 0, I_2 \mapsto 0\}, \{O_L \mapsto (t_3,a_3)\})}_x$. It remains to prove that $r_n = \mu_{id}([Q])$. Let us notice that the derivation of $P \trs{} \mu$ is actually as follows:
{\small
\begin{prooftree}
	\AxiomC{$\coexo_n.\tau \trs{} [\coexo_n]_{\tau}$}
	\AxiomC{$\tau_0 \trs{} \mu$}
	\LeftLabel{\!\!\!\!\!(par)}
	\BinaryInfC{$\coexo_n.\tau|\tau_0 \trs{} [\coexo_n]_{\tau} \mpar{\coexo_n.\tau}{\tau_0} \mu$}
	\AxiomC{$\exo_n.\sigma \trs{} [\exo_n]_\sigma$}
	\AxiomC{$\sigma_0 \trs{} \mu''$}
	\RightLabel{(par)}
	\BinaryInfC{$\exo_n.\sigma|\sigma_0 \trs{} [\exo_n]_{\sigma} \mpar{\exo_n.\sigma}{\sigma_0} \mu''$}
	\AxiomC{$P' \trs{} \mu'$}
	\RightLabel{(loc)}
	\BinaryInfC{$\cell{\exo.\sigma|\sigma_0}{P'} \trs{} \mu' \mAt{\exo_n.\sigma|\sigma_0}{P'} ([\exo_n]_\sigma\mpar{\exo_n.\sigma}{\sigma_0}\mu'')$}
	\AxiomC{$P'' \trs{} \mu'''$}
	\LeftLabel{\!\!\!\!\!\!\!(comp)}
	\BinaryInfC{\!\!\!\!\!\!\!$\cell{\exo_n.\sigma|\sigma_0}{P'} \comp P'' \trs{} 
	(\mu' \mAt{\exo_n.\sigma|\sigma_0}{P'} ([\exo_n]_\sigma\mpar{\exo_n.\sigma}{\sigma_0}\mu'')) 
	\mtimes{\cell{\exo_n.\sigma|\sigma_0}{P'}}{P''} \mu'''$}
	\LeftLabel{\!\!\!\!\!(loc)}
	\BinaryInfC{$\cell{\coexo_n.\tau|\tau_0}{\cell{\exo_n.\sigma|\sigma_0}{P'} \comp P''} \trs{} \nu$}
\end{prooftree}
}
\noindent 
where $\nu = ((\mu' \mAt{\exo_n.\sigma|\sigma_0}{P'} ([\exo_n]_\sigma\mpar{\exo_n.\sigma}{\sigma_0}\mu'')) 
	\mtimes{\cell{\exo_n.\sigma|\sigma_0}{P'}}{S} \mu''') 
	\mAt{\coexo_n.\tau|\tau_0}{\cell{\exo_n.\sigma|\sigma_0}{P'} \comp P''} 
	([\coexo_n]_{\tau} \mpar{\coexo_n.\tau}{\tau_0} \mu)$, \\
$\mu_1 = (\mu' \mAt{\exo_n.\sigma|\sigma_0}{P'}
([\exo_n]_\sigma\mpar{\exo_n.\sigma}{\sigma_0}\mu''))
\mtimes{\cell{\exo_n.\sigma|\sigma_0}{P'}}{P''} \mu'''$ and $\mu_2 =
[\coexo_n]_{\tau} \mpar{\coexo_n.\tau}{\tau_0} \mu$.
Then:
\renewcommand{\frac}[2]{#1/#2}
\begin{align*}
&\nu_{id}([\cell{\sigma|\sigma_0|\tau|\tau_0}{P''} \comp P']) 
=	(\mu_1~ \mAt{\coexo_n.\tau|\tau_0}{\cell{\exo_n.\sigma|\sigma_0}{P'} \comp P''}~ \mu_2)_{id}([\cell{\sigma|\sigma_0|\tau|\tau_0}{P''} \comp P'])
	\\
&= {\mu_1}_{id}([\cell{\sigma|\sigma_0|\tau|\tau_0}{P''} \comp P']) + 
\frac{{\mu_1}_{ex_n}([\sigma|\sigma_0] \times [P'] \times [P'']) \cdot {\mu_2}_{\coexo_n} ([\tau|\tau_0])}{r_n}\\
&= \frac{{\mu_1}_{ex_n}([\sigma|\sigma_0] \times [P'] \times [P'']) \cdot {\mu_2}_{\coexo_n} ([\tau|\tau_0])}{r_n}\\
&= \frac{({(\mu' \mAt{\exo_n.\sigma|\sigma_0}{P'} ([\exo_n]_\sigma\mpar{\exo_n.\sigma}{\sigma_0}\mu''))}_{ex_n}([\sigma|\sigma_0] \times [P'] \times [P'']) +
	{\mu'''_{ex_n}([\sigma|\sigma_0] \times [P'] \times [P''])})
\cdot {\mu_2}_{\coexo_n} ([\tau|\tau_0])}{r_n}
\\
&= \frac{( ([\exo_n]_{\sigma}~ \mpar{\exo_n.\sigma}{\sigma_0}~ \mu'')_{\exo_n}([\sigma|\sigma_0]) +
	{\mu'''_{ex_n}([\sigma|\sigma_0] \times [P'] \times [P''])}) \cdot
	{([\coexo_n]_{\tau} \mpar{\coexo_n.\tau}{\tau_0} \mu)}_{\coexo_n} ([\tau|\tau_0])}{r_n}\\
&= \frac{( ([\exo_n]_{\sigma})_{\exo_n} ([\sigma|\sigma_0]) + \mu''_{\exo_n} ([\sigma|\sigma_0]) +
 	{\mu'''_{ex_n}([\sigma|\sigma_0] \times [P'] \times [P''])}) \cdot
 	(([\coexo_n]_{\tau})_{\coexo_n} ([\tau|\tau_0]) + \mu_{\coexo_n} ([\tau|\tau_0]) )
 	}{r_n}\\
&= \frac{(r_n + \mu''_{\exo_n} ([\sigma|\sigma_0]) + {\mu'''_{ex_n}([\sigma|\sigma_0] \times [P'] \times [P''])} ) \cdot
	(r_n + \mu_{\coexo_n} ([\tau|\tau_0]) )}{r_n} = r_n
\end{align*}

where the last equivalence holds because $\mu''_{\exo_n} ([\sigma|\sigma_0]) = \mu'''_{ex_n}([\sigma|\sigma_0] \times [P'] \times [P'']) = \mu_{\coexo_n} ([\tau|\tau_0]) = 0$ because we assumed that the reaction does not involve either $\sigma_0$ nor $\tau_0$.

\end{document}